\documentclass[11pt]{article}
\usepackage{fullpage}
\usepackage{graphicx}
\usepackage{amssymb}
\usepackage{amsmath}
\usepackage{amsthm}
\usepackage[table]{xcolor}
\usepackage{stmaryrd}
\usepackage{mathabx}
\usepackage{mathrsfs}

\usepackage{latexsym}
\usepackage{rotating}

\usepackage{enumitem}

\newtheorem{theorem}{Theorem}[section]
\newtheorem{lemma}[theorem]{Lemma}

\theoremstyle{definition}

\newtheorem{remark}[theorem]{Remark}
             
\newcommand{\ignore}[1]{}
             
\newcommand{\hcm}[1][1]{\hspace*{#1 cm}}
\newcommand{\rb}[2]{\raisebox{#1 mm}[0mm][0mm]{#2}}
\newcommand{\istrut}[2][0]{\rule[- #1 mm]{0mm}{#1 mm}\rule{0mm}{#2 mm}}

\newcommand{\zero}[1]{\makebox[0mm][r]{#1}}

\newcommand{\E}{{\mathbb E\/}}

\newcommand{\ceil}[1]{\lceil #1 \rceil}
\newcommand{\floor}[1]{\lfloor #1 \rfloor}
\newcommand{\f}[2]{\frac{#1}{#2}}

\newcommand{\bydef}{\stackrel{\operatorname{def}}{=}}
\newcommand{\poly}{\operatorname{poly}}
\newcommand{\polylog}{\operatorname{polylog}}
\newcommand{\bottom}{\perp}
\newcommand{\argmin}{\operatornamewithlimits{argmin}}

\newcommand{\lo}{\operatorname{lo}}
\newcommand{\hi}{\operatorname{hi}}
\newcommand{\targetplus}{\ocirc}
\newcommand{\colorfunc}{\mbox{\sc color}}
\newcommand{\RAND}{{\sc rand.}}
\newcommand{\DECTREE}{{\sc dec.~tree}}
\newcommand{\contour}{\mbox{\sc contour}}

\newcommand{\row}{\operatorname{r}}
\newcommand{\col}{\operatorname{c}}

\newcommand{\ThreeSUM}{\textsf{3SUM}}
\newcommand{\TwoSUM}{\textsf{2SUM}}
\newcommand{\SUM}{\textsf{SUM}}
\newcommand{\ConvolutionThreeSUM}{\textsf{Convolution3SUM}}
\newcommand{\IntegerConvolutionThreeSUM}{\textsf{IntegerConv3SUM}}
\newcommand{\LDT}{\textsf{LDT}}
\newcommand{\IntegerThreeSUM}{\textsf{Integer3SUM}}
\newcommand{\ZeroTriangle}{\textsf{ZeroTriangle}}

\newcommand{\LoveTriangle}{\textsf{LoveTriangle}}

\newcommand{\Erdos}{Erd\H{o}s}

\newcommand{\Patrascu}{P\v{a}tra\c{s}cu}
\newcommand{\Turan}{Tur\'{a}n}

\title{Threesomes, Degenerates, and Love Triangles\thanks{This work is supported in part by the Danish National Research Foundation
grant DNRF84 through the Center for Massive Data Algorithmics (MADALGO).
S. Pettie is supported by NSF grants CCF-1217338 and CNS-1318294
and a grant from the US-Israel Binational Science Foundation.}}

\author{Allan Gr\o nlund \\ MADALGO, Aarhus University \and Seth Pettie\\ University of Michigan}
\begin{document}

\maketitle

\begin{abstract}
The \ThreeSUM{} problem is to decide, given a set of $n$ real numbers, whether
any three sum to zero.  It is widely conjectured that a trivial $O(n^2)$-time algorithm is optimal
and over the years the consequences of this conjecture have been revealed.
This \ThreeSUM{} conjecture implies $\Omega(n^2)$ lower bounds on numerous problems in computational geometry
and a variant of the conjecture
implies strong lower bounds on triangle enumeration, dynamic graph algorithms, and string matching data structures.
 
In this paper we refute the \ThreeSUM{} conjecture.
We prove that the decision tree complexity of \ThreeSUM{} is $O(n^{3/2}\sqrt{\log n})$
and give two subquadratic \ThreeSUM{} algorithms, a deterministic one running in $O(n^2 / (\log n/\log\log n)^{2/3})$ time
and a randomized one running in $O(n^2 (\log\log n)^2 / \log n)$ time with high probability.  
Our results lead directly to improved bounds for
$k$-variate linear degeneracy testing for all odd $k\ge 3$.
The problem is to decide, given a linear function 
$f(x_1,\ldots,x_k) = \alpha_0 + \sum_{1\le i\le k} \alpha_i x_i$ and a set $A \subset \mathbb{R}$,
whether $0\in f(A^k)$.
We show the decision tree complexity of this problem is $O(n^{k/2}\sqrt{\log n})$.

Finally, we give a subcubic algorithm for a generalization of the $(\min,+)$-product over real-valued matrices and apply it
to the problem of finding zero-weight triangles in weighted graphs.
We give a depth-$O(n^{5/2}\sqrt{\log n})$ decision tree for this problem, 
as well as an algorithm running in time $O(n^3 (\log\log n)^2/\log n)$.
\end{abstract}

\section{Introduction}

The time hierarchy theorem~\cite{HartmanisS65} implies that there exist problems in $\mathbf{P}$ with complexity $\Omega(n^k)$ for every fixed $k$.
However, it is consistent with current knowledge that all problems of practical interest can be solved in 
$\tilde{O}(n)$ time in a reasonable model of computation.
Efforts to build a useful complexity theory inside $\mathbf{P}$ have been based on the conjectured hardness
of certain archetypal problems, such as \ThreeSUM, $(\min,+)$-matrix product, and CNF-SAT.
See, for example, the conditional lower bounds in \cite{GajentaanO95,Patrascu10,PatrascuW10,JafargholiV13,AbboudL13,AbboudW14,RodittyW13,ChechikLRSTW14,WilliamsW10}.

In this paper we study the complexity of \ThreeSUM{} and related problems such as linear degeneracy testing (\LDT)
and finding zero-weight triangles.  Let us define the problems formally.

\setdescription{leftmargin=1.5cm,labelindent=\parindent}
\begin{description}
\item[{\bf \ThreeSUM:}]
Given a set $A \subset \mathbb{R}$, determine if there exists $a,b,c\in A$ such 
that $a+b+c=0$.

\item[{\bf \IntegerThreeSUM:}] Given a set $A \subseteq \{-U,\ldots,U\}\subset \mathbb{Z}$, determine if there exists $a,b,c\in A$
such that $a+b+c=0$.

\item[{\bf $k$-\LDT{} and $k$-\SUM:}]
Fix a $k$-variate linear function $\phi(x_1,\ldots,x_k) = \alpha_0 + \sum_{i=1}^{k} \alpha_i x_i$, where $\alpha_0,\ldots,\alpha_k \in \mathbb{R}$.
Given a set $A \subset \mathbb{R}$, determine if $\phi(\mathbf{x})=0$ for any $\mathbf{x}\in A^k$.
When $\phi$ is $\sum_{i=1}^k x_i$ the problem is called $k$-\SUM.  

\item[{\bf \ZeroTriangle:}]
Given a weighted undirected graph $G=(V,E,w)$, where $w : E\rightarrow \mathbb{R}$, determine if there exists a triangle $(a,b,c)\in V^3$
for which $w(a,b)+w(b,c)+w(c,a) =0$.  
(From the definition of {\em love : a score of zero}, one could also call this the \LoveTriangle{} problem.)

\end{description}

These problems are often defined with further constraints that do not change the problem in any substantive way~\cite{GajentaanO95}.
For example, the input to \ThreeSUM{} can be three sets $A,B,C\subset \mathbb{R}$ and the problem is to determine if
there exists $a\in A, b\in B, c\in C$ such that $a+b+c=0$.  Even if there is only one set, there is sometimes an additional constraint
that $a,b,$ and $c$ be {\em distinct} elements.

As a problem in its own right, \ThreeSUM{} has no compelling practical applications.  However, lower bounds on \ThreeSUM{} imply lower
bounds on dozens of other problems that {\em are} of practical interest.  Before reviewing existing \ThreeSUM{} algorithms we
give a brief survey of conditional lower bounds that depend on the hardness of \ThreeSUM.

\subsection{Implications of the \ThreeSUM{} Conjectures}\label{sect:review-ThreeSUMhard}

It is often conjectured that \ThreeSUM{} requires $\Omega(n^2)$ time and that \IntegerThreeSUM{} requires $\Omega(n^{2-o(1)})$ time~\cite{Patrascu10,AbboudW14}.
These conjectures have been shown to imply strong lower bounds on numerous problems in computational
geometry~\cite{GajentaanO95,AichholzerADHRU12,BarequetH01,SossEO03}
dynamic graph algorithms~\cite{Patrascu10,AbboudW14},
and pattern matching~\cite{AbboudWW14,AmirCLL14,ButmanCCJLPPS13,ChenHC09}.
For example, 
the \ThreeSUM{} conjecture implies that the following problems require at least $\Omega(n^2)$ time.
\begin{enumerate}
\item[---] Given an $n$-point set in $\mathbb{R}^2$, determine whether it contains three collinear points~\cite{GajentaanO95}.
\item[---] Given two $n$-edge convex polygons,  determine whether one can be placed inside the other via rotation and translation~\cite{BarequetH01}.
\item[---] Given $n$ triangles in $\mathbb{R}^2$, determine whether their union contains a hole, or determine the area of their union~\cite{GajentaanO95}.
\end{enumerate}
Through a series of reductions, \Patrascu~\cite{Patrascu10} proved that the \IntegerThreeSUM{} conjecture implies lower bounds on triangle enumeration and
various problems in dynamic data structures, even when all updates and queries are presented in advance.
Some lower bounds implied by the \IntegerThreeSUM{} conjecture include the following.
\begin{enumerate}
\item[---] Given an undirected $m$-edge graph, enumerating up to $m$ triangles (3-cycles) requires at least $\Omega(m^{4/3-o(1)})$ time~\cite{Patrascu10}.\footnote{Bjorklund et al.~\cite{BjorklundPWZ14} recently proved that the exponent 
$4/3$ is optimal if the matrix multiplication exponent $\omega$ is 2.}
\item[---] Given a sequence of $m$ updates to a directed graph (edge insertions and edge deletions) and two specified vertices $s,t$,
determining whether $t$ is reachable from $s$ after each update requires at least $\Omega(m^{4/3-o(1)})$ time~\cite{AbboudW14}.
\item[---] Given an edge-weighted undirected graph, deciding whether there exists a zero-weight triangle requires at least $\Omega(n^{3-o(1)})$ time~\cite{WilliamsW13}.
\end{enumerate}

In recent years conditional lower bounds have been obtained from two other plausible conjectures: that computing the $(\min,+)$-product of two $n\times n$ matrices
takes $\Omega(n^{3-o(1)})$ time and that 
CNF-SAT takes $\Omega(2^{(1-o(1))n})$ time.  The latter is sometimes called the Strong Exponential Time Hypothesis (Strong ETH).
We now know that if the Strong ETH holds, no $n^{o(k)}$ algorithm exists for $k$-\SUM~\cite{PatrascuW10}
and no $m^{2-\Omega(1)}$ algorithm exists for $(3/2-\epsilon)$-approximating the diameter of an $m$-edge unweighted graph~\cite{RodittyW13,ChechikLRSTW14}.
Williams and Williams~\cite{WilliamsW10} proved that numerous problems are equivalent to $(\min,+)$-matrix multiplication, inasmuch as a truly subcubic ($O(n^{3-\Omega(1)})$)
algorithm for one would imply truly subcubic algorithms for all the others.

\subsection{Algorithms, Lower Bounds, and Reductions}\label{sect:review-ThreeSUM-algs}

The evidence in favor of the \ThreeSUM{} and \IntegerThreeSUM{} conjectures is rather thin.
Erickson~\cite{Erickson99} and Ailon and Chazelle~\cite{AilonC05} proved that any $k$-linear decision tree for solving $k$-\LDT{} must have 
depth $\Omega(n^{k/2})$ when $k$ is even and $\Omega(n^{(k+1)/2})$ when $k$ is odd.
In particular, any 3-linear decision tree for \ThreeSUM{} has depth $\Omega(n^2)$.  (An $s$-linear decision tree is one where each internal node
asks for the sign of a linear expression in $s$ elements.)
The \IntegerThreeSUM{} problem is obviously not harder than \ThreeSUM, but no other relationship between these two problems is known.
Indeed, the assumption that elements are integers opens the door to a variety of algorithmic techniques that cannot be modeled as decision trees.
Using the fast Fourier transform it is possible to solve \IntegerThreeSUM{} in 
$O(n + U\log U)$ time, which is subquadratic even for a rather large universe size $U$.\footnote{Erickson~\cite{Erickson99} 
credits R.~Seidel with this \ThreeSUM{} algorithm.}
Baran, Demaine, and \Patrascu~\cite{BaranDP08} showed that \IntegerThreeSUM{} can be solved in $O(n^2 / (\log n/\log\log n)^2)$ time (with high probability)
on the word RAM, where $U = 2^w$ and $w>\log n$ is the machine word size.
The algorithm uses a mixture of randomized universe reduction (via hashing), word packing, and table lookups.

It is straightforward to reduce $k$-\LDT{} to a \TwoSUM{} problem or unbalanced \ThreeSUM{} problem, depending on whether $k$ is even or odd.
When $k$ is odd one forms certain sets $A,B,C$ where $|A|=|B|=n^{(k-1)/2}$ and $|C|=n$, then sorts them in $O(n^{(k-1)/2}\log n)$ time.  
The standard three-set \ThreeSUM{} algorithm on $A,B,C$ takes $O(|C|(|A|+|B|)) = O(n^{(k+1)/2})$ time.  When $k\ge 4$ is even there is no set $C$.
Using Lambert's algorithm~\cite{Lambert92}, $A$ and $B$ can be sorted is $O(n^{k/2}\log n)$ time while performing only $O(n^{k/2})$ comparisons.
These algorithms can be modeled as $k$-linear decision trees, and are therefore optimal in this model by the lower bounds of~\cite{Erickson99,AilonC05}.
However, it was known that all $k$-\LDT{} problems can be solved by $n$-linear decision trees with depth $O(n^5\log n)$~\cite{Meiser93}, 
or with depth $O(n^4\log(nK))$ if the coefficients of the linear function are integers with absolute value at most $K$~\cite{Heide84}.
Unfortunately these decision trees are not efficiently constructible.
The time required to determine {\em which} comparisons to make is exponential in $n$.

The \ZeroTriangle{} problem was highlighted in a recent article by Williams and Williams~\cite{WilliamsW13}.  They
did not give any subcubic algorithm, but did show that a subcubic \ZeroTriangle{} algorithm would have implications 
for \IntegerThreeSUM{} via an intermediate problem called \ConvolutionThreeSUM.
\setdescription{leftmargin=1.5cm,labelindent=\parindent}
\begin{description}
\item[{\bf \ConvolutionThreeSUM:}] Given a vector $A \in \mathbb{R}^n$, determine if there exist $i,j$ for which $A(i)+A(j) = A(i+j)$.
\item[{\bf \IntegerConvolutionThreeSUM:}] The same as \ConvolutionThreeSUM, except that
$A \in \{0,\ldots,U-1\}^n$ and $U\le 2^w$, where $w=\Omega(\log n)$ is the machine word size.
\end{description}

\Patrascu~\cite{Patrascu10} defined the \ConvolutionThreeSUM{} problem and gave a randomized reduction from \IntegerThreeSUM{} to 
\IntegerConvolutionThreeSUM.  Williams and Williams~\cite{WilliamsW13} gave a reduction from \ConvolutionThreeSUM{} to \ZeroTriangle.
Neither of these reductions is frictionless.  Define $T_{I3S},T_{IC3S},T_{C3S}$ and $T_{ZT}$ to be the complexities of the various problems on 
inputs of length $n$, or graphs with $n$ vertices.  Clearly $T_{IC3S}(n) \le T_{C3S}(n)$.  The reductions show that for any $k$,  
$T_{I3S}(n) = O(n^2/k + k^3\cdot T_{IC3S}(n/k))$ and $T_{C3S}(n) = O(\sqrt{n}\cdot T_{ZT}(\sqrt{n}))$.  
Note that even if \ZeroTriangle{} had an $O(n^2)$-time algorithm (optimal on dense graphs),
this would only give an $O(n^{9/5})$ bound for \IntegerThreeSUM.

\begin{table}
\centering
\scalebox{1}{
\begin{tabular}{r|l@{\hcm[1.8]\istrut{5}}rcr|l@{\hcm[1.8]}r}
\multicolumn{3}{l}{\ThreeSUM}								&\hcm[.4] &	\multicolumn{3}{l}{\IntegerThreeSUM}\\\cline{1-3}\cline{5-7}
trivial		& $n^2$ 		&								&&	trivial		& $n^2$ & \\
		& 							&				&&	Seidel 1997	& $n + U\log U$ &	\\
		& \rb{4}{$n^{3/2}\sqrt{\log n}\;$} & \rb{4}{\zero{\DECTREE}} 						&&	\rb{-1}{Baran, Demaine,} 	& \rb{-1}{$n^2 / \left(\f{\log n}{\log\log n}\right)^2$} & \rb{-1}{\zero{\RAND}}\\
\rb{2.5}{\bf new}	& \rb{2.5}{$n^2 / \left(\frac{\log n}{\log\log n}\right)^{2/3}$} 	&		&&	\rb{0}{\hcm[.1] \Patrascu{} 2005} &	$n^2/ \frac{w}{\log^2 w}$ &\zero{\RAND}\\
		& \rb{2}{$n^2 / \frac{\log n}{(\log\log n)^2}$} & \rb{2}{\zero{\RAND}}	&&	{\bf new}		& \istrut[3]{0}$n^2 / \left(\frac{\log n}{\log\log n}\right)^{2/3}$ & \\\cline{1-3}\cline{5-7}
\multicolumn{7}{l}{ }\\
\multicolumn{7}{l}{ }\\
\multicolumn{3}{l}{\ConvolutionThreeSUM}					&&	\multicolumn{3}{l}{\ZeroTriangle}\\\cline{1-3}\cline{5-7}
trivial		& $n^2$ &										&&	trivial		& $n^3$&\\
		& $n^{3/2}\sqrt{\log n}\;$ &\zero{\DECTREE}					&& 			& $n^{5/2}\sqrt{\log n}\;$ &\zero{\DECTREE}\\			
		& $n^{3/2}\;$  & \zero{\RAND, \DECTREE}					&& 			& $n^{5/2}\;$ &\zero{\RAND, \DECTREE}\\
\rb{3}{\bf new}		& $n^2 / \frac{\log n}{(\log\log n)^2}$	&				&&			& $n^3 / \frac{\log n}{(\log\log n)^2}$ &\\
		& $n^2 / \frac{\log n}{\log\log n}\;$ & \istrut[3]{0}\zero{\RAND}			&&			& \istrut[2]{0}$n^3 / \frac{\log n}{\log\log n}\;$ &\zero{\RAND}\\\cline{1-3}
\multicolumn{3}{l}{ }										&&		\rb{3}{\bf new} & \istrut[2]{0}$m^{5/4}\sqrt{\log m}\;$ &\zero{\DECTREE}\\
\multicolumn{3}{l}{ }										&&			& $m^{5/4}$ & \zero{\RAND, \DECTREE}\\
\multicolumn{3}{l}{ }										&&			& $m^{3/2} / \left(\frac{\log m}{(\log\log m)^2}\right)^{1/4}$ &\\
\multicolumn{3}{l}{ }										&&			& $m^{3/2} / \left(\frac{\log m}{\log\log m}\right)^{1/4}\;$ &\zero{\RAND}\\\cline{5-7}
\end{tabular}
}
\caption{\label{fig:results} A summary of the results.
Results in the decision tree model are indicated by \DECTREE,  results using randomization are indicated by \RAND{}
In \ZeroTriangle{} $n$ and $m$ are the number of vertices and edges whereas in all other problems $n$ is the length of the input.
In \IntegerThreeSUM{} $w=\Omega(\log n)$ is the machine word size and $U \le 2^w$ the size of the universe.
}
\end{table}

\subsection{New Results.}

We give the first subquadratic bounds on both the decision tree complexity of \ThreeSUM{} and the algorithmic complexity of \ThreeSUM,
which also gives the first deterministic subquadratic algorithm for \IntegerThreeSUM.\footnote{We assume a simplified 
Real RAM model.  Real numbers are subject to only two unit-time operations: addition and comparison.
In all other respects the machine behaves like a $w=O(\log n)$-bit word RAM with the standard repertoire of unit-time 
$AC^0$ operations: bitwise Boolean operations, left and right shifts, addition, and comparison.}
Our method leads to similar improvements to the decision tree complexity of $k$-\LDT when $k\ge 3$ is odd.
Refer to Figure~\ref{fig:results} for a summary of prior work and our results.

\begin{theorem}\label{thm:3SUM}
There is a 4-linear decision tree for \ThreeSUM{} with depth $O(n^{3/2}\sqrt{\log n})$.
Furthermore, \ThreeSUM{} can be solved deterministically in $O(n^2 / (\log n/\log\log n)^{2/3})$ time
and, using randomization, in $O(n^2 (\log \log n)^2/ \log n)$ time with high probability.
\end{theorem}

\begin{theorem}\label{thm:kLDT}
When $k\ge 3$ is odd, there is a $(2k-2)$-linear decision tree for $k$-\LDT{} with depth $O(n^{k/2}\sqrt{\log n})$, 
\end{theorem}

Theorem~\ref{thm:3SUM} refutes the \ThreeSUM{} conjecture and
casts serious doubts on the optimality of many $O(n^2)$ algorithms in computational geometry.
Theorem~\ref{thm:3SUM} also answers a question of Erickson~\cite{Erickson99} and Ailon and Chazelle~\cite{AilonC05} about whether $(k+1)$-linear decision
trees are more powerful than $k$-linear decision trees in solving $k$-\LDT{} problems.  
In the case of $k=3$, they are.

We define a new product of three real-valued matrices called {\em target-min-plus}, which is trivially computable in $O(n^3)$ time.
We observe that \ZeroTriangle{} is reducible to a target-min-plus product, then give subcubic bounds on the decision tree and
algorithmic complexity of target-min-plus.  Theorem~\ref{thm:ZeroTriangle} is an immediate consequence.

\begin{theorem}\label{thm:ZeroTriangle}
The decision tree complexity of \ZeroTriangle{}
is $O(n^{5/2}\sqrt{\log n})$ on $n$-vertex graphs and its randomized decision tree complexity
is $O(n^{5/2})$ with high probability.
There is a deterministic \ZeroTriangle{} algorithm running in $O(n^3 (\log\log n)^2 / \log n)$ time
and a randomized algorithm running in $O(n^3 \log\log n/\log n)$ time with high probability.
\end{theorem}

Any $m$-edge graph contains $O(m^{3/2})$ triangles which can be enumerated in
$O(m^{3/2})$ time, so \ZeroTriangle{} can clearly be solved in $O(m^{3/2})$ time as well.
We improve this bound for all $m$.  

\newcommand{\StatementthmsparseZeroTriangle}{
The decision tree complexity of \ZeroTriangle{} on $m$-edge graphs is
$O(m^{5/4}\sqrt{\log m})$ and, using randomization, $O(m^{5/4})$ with high probability.
The \ZeroTriangle{} problem can be solved in 
$O(m^{3/2} (\log\log m)^2/\log m)$ time deterministically
or $O(m^{3/2}\log\log m/\log m)$ with high probability.}
\begin{theorem}\label{thm:sparse-ZeroTriangle}
\StatementthmsparseZeroTriangle
\end{theorem}

By invoking the Williams-Williams reduction~\cite{WilliamsW13}, our \ZeroTriangle{} algorithms give subquadratic bounds
on the complexity of \ConvolutionThreeSUM.  By designing \ConvolutionThreeSUM{} algorithms from scratch we can
obtain speedups comparable to those of Theorem~\ref{thm:ZeroTriangle}.

\newcommand{\StatementthmConvolutionThreeSUM}{The decision tree complexity of \ConvolutionThreeSUM{} is $O(n^{3/2}\sqrt{\log n})$
and its randomized decision tree complexity is $O(n^{3/2})$ with high probability.
The \ConvolutionThreeSUM{} problem can be solved  
in $O(n^2 (\log\log n)^2 / \log n)$ time deterministically, or in $O(n^2 \log\log n/\log n)$ time with high probability.}
\begin{theorem}\label{thm:ConvolutionThreeSUM}
\StatementthmConvolutionThreeSUM
\end{theorem}

\subsection{An Overview}

All of our algorithms borrow liberally from Fredman's 1976 articles on the decision tree complexity of $(\min,+)$-matrix multiplication~\cite{F76}
and the complexity of sorting $X+Y$~\cite{Fredman76}.  Throughout the paper we shall refer to the ingenious observation 
that $a+b < c+d$ iff $a-c < d-b$ as {\em Fredman's trick}.\footnote{Noga Alon (personal communication) remarked that this trick dates
back to \Erdos{} and \Turan~\cite{ErdosT41}, {\em if not further!}}
In order to shave off $\poly(\log n)$ factors in runtime we apply the
geometric domination technique invented by Chan~\cite{Chan08}
and developed further by 
Bremner, Chan, Demaine, Erickson, Hurtado, Iacono, 
Langerman, \Patrascu, and Taslakian~\cite{BremnerCDEHILPT14}.

In Section~\ref{sect:usefullemmas} we review a number of useful lemmas due to Fredman~\cite{Fredman76}, Buck~\cite{Buck43},
and Chan~\cite{Chan08} about sorting with partial information,
the complexity of hyperplane arrangements, and the complexity of dominance reporting in $\mathbb{R}^d$.
In Section~\ref{sect:quadratic3SUM} we review a standard $O(n^2)$-time \ThreeSUM{} algorithm
and in Section~\ref{sect:nonuniform} we present an $\tilde{O}(n^{3/2})$-depth decision tree for \ThreeSUM.
Subquadratic algorithms for \ThreeSUM{} are presented in Section~\ref{sect:uniform}.
Section~\ref{sect:lineardegeneracy} presents new bounds on the decision tree complexity of $k$-\LDT{} for odd $k\ge 3$.
Section~\ref{sect:ZeroTriangle} presents new bounds on the decision tree and algorithmic complexity of \ZeroTriangle{}
and \ConvolutionThreeSUM. 
Section~\ref{sect:conclusion} concludes with some open problems.

\section{Useful Lemmas}\label{sect:usefullemmas}

Fredman~\cite{Fredman76} considered the problem of sorting a list of $n$ numbers known to be arranged in one of
$\Pi \le n!$ permutations.
When $\Pi$ is sufficiently small the list can be sorted using a linear number of comparisons.

\begin{lemma}\label{lem:Fredman} (Fredman 1976~\cite{Fredman76})
A list of $n$ numbers whose sorted order is one of $\Pi$ permutations can be sorted with
$2n+\log\Pi$ pairwise comparisons.
\end{lemma}

Throughout the paper $[N]$ denotes the first $\ceil{N}$ natural numbers $\{0,\ldots,\ceil{N}-1\}$, where $N$ may or may not be an integer.
We apply Lemma~\ref{lem:Fredman} to the problem of sorting {\em Cartesian sums}.  
Given lists $A = (a_i)_{i\in[n]}$ and $B=(b_i)_{i\in [n]}$ of distinct numbers, 
define
$A+B = \{a_i + b_j \:|\: i,j \in [n]\}$.
We often regard $A+B$ as an $|A|\times |B|$ matrix (which may contain multiple copies of the same number) 
or as a point in the $2n$-dimensional space $\mathbb{R}^{2n}$, whose coordinates are named $x_1,\ldots,x_n,y_1,\ldots,y_n$.
The points in $\mathbb{R}^{2n}$ that agree with a fixed permutation of $A+B$ form a convex cone bounded 
by the ${n^2 \choose 2}$ hyperplanes $H = \{x_i + y_j - x_k - y_l \:|\: i,j,k,l\in[n]\}$.
The sorted order of $A+B$ is encoded as a sign vector $\{-1,0,1\}^{n^2\choose 2}$
depending on whether $(A,B)$ lies on, above, or below a particular hyperplane in $H$.  Therefore the number
of possible sorted orders of $A+B$ is exactly the number of regions (of all dimensions) defined by the arrangement $H$.
(Regions of dimension less than $2n$ correspond to instances in which some numbers appear multiple times.)

\begin{lemma}\label{lem:Buck}
(Buck 1943~\cite{Buck43})
Consider the partition of space defined by an arrangement of $m$ hyperplanes in $\mathbb{R}^d$.
The number of regions of dimension $k\le d$ is at most
\[
{m\choose d-k}\left({m - d+k \choose 0} + {m - d+k \choose 1} + \cdots + {m - d+k \choose k}\right)
\]
and the number of regions of all dimensions is $O(m^d)$.
\end{lemma}

In one of our algorithms we will construct the hyperplane arrangement explicitly.
Edelsbrunner, O'Rourke, and Seidel~\cite{EdelsbrunnerOS86} proved that the natural
incremental algorithm takes $O(m^d)$ time (linear in the size of the arrangement),
but any trivial $m^{O(d)}$-time algorithm suffices in our application.
The hyperplane arrangements we use correspond to fragments of the Cartesian sum $A+B$.
Lemma~\ref{lem:sortX+Y} is a direct consequence of Lemmas~\ref{lem:Fredman} and \ref{lem:Buck}.

\begin{lemma}\label{lem:sortX+Y}
Let $A=(a_i)_{i\in[n]}$ and $B=(b_i)_{i\in [n]}$ be two lists of numbers
and let $F\subseteq [n]^2$ be a set of positions in the $n\times n$ grid.
The number of realizable orders of
$(A+B)_{|F} \bydef \{a_i + b_j \:|\: (i,j) \in F\}$ is $O\mathopen{}\left({|F|\choose 2}^{2n}\right)\mathclose{}$
and therefore $(A+B)_{|F}$ can be sorted with at most $2|F| + 2n\log|F| + O(1)$ comparisons.
\end{lemma}

It is sometimes convenient to assume that the elements of a Cartesian sum
are distinct (and therefore have exactly one sorted order), even though numbers 
may appear multiple times.  Lemma~\ref{lem:tiebreaking} illustrates one way to break ties consistently.
The proof is straightforward.

\begin{lemma}\label{lem:tiebreaking}
Let $A=(a_i)$ and $B=(b_i)$ be two lists of numbers.  Define
$a_i' = (a_i, i, 0)$ and $b_j' = (b_j,0,j)$.
The Cartesian sum $A' + B'$ is totally ordered, and is 
a linear extension of the partially ordered $A+B$.
(Addition over tuples is pointwise addition;
tuples are ordered lexicographically.
The tuple $(u,v,w)$ can be regarded as a representation of a real number $u + \epsilon_1 v + \epsilon_2 w$ where $\epsilon_1 \gg \epsilon_2$
are sufficiently small so as not to invert strictly ordered elements of $A+B$.)
\end{lemma}

Given a set $P$ of red and blue points in $\mathbb{R}^d$, 
the {\em bichromatic dominating pairs} problem is 
to enumerate every pair $(p,q)\in P^2$ such that $p$ is red, $q$ is blue,
and $p$ is greater than $q$ at each of the $d$ coordinates.
A natural divide and conquer algorithm~\cite[p.~366]{PreparataShamos85} runs in time linear in the output size and 
$O(n\log^d n)$.  Chan~\cite{Chan08} provided an improved analysis when $d$ is logarithmic in $n$.
For the sake of completeness we give a short proof of Lemma~\ref{lem:redblue} in Appendix~\ref{sect:redblue}.

\begin{lemma}\label{lem:redblue} {\bf (Bichromatic Dominance Reporting~\cite{Chan08})}
Given a set $P\subseteq \mathbb{R}^d$ of red and blue points, it is
possible to return all bichromatic dominating pairs $(p,q)\in P^2$
in time linear in the output size and $c_\epsilon^d |P|^{1+\epsilon}$.
Here $\epsilon\in (0,1)$ is arbitrary and $c_\epsilon = 2^\epsilon/(2^{\epsilon}-1)$.
\end{lemma}

We typically invoke Lemma~\ref{lem:redblue} with $\epsilon=1/2, c_\epsilon \approx 3.42,$ and
$d = \delta \log n$, where $\delta >0$ is sufficiently small
to make the running time subquadratic, excluding the time allotted to reporting the output.

\section{The Quadratic \ThreeSUM{} Algorithm}\label{sect:quadratic3SUM}

We shall review a standard $O(n^2)$ algorithm for the three-set version of \ThreeSUM{} and introduce some terminology used in
Sections~\ref{sect:nonuniform} and \ref{sect:uniform}.
We are given sets $A,B,C \subset \mathbb{R}$ and must determine if there exists $a\in A, b\in B, c\in C$ such
that $a+b+c=0$.  For each $c\in C$ the algorithm searches for $-c$ in the Cartesian sum $A+B$.
Each search takes $O(|A|+|B|)$ time, for a total of $O(|C|(|A|+|B|))$.  We view $A+B$ as being a matrix
whose rows correspond to $A$ and columns correspond to $B$, both listed in increasing order.

\setdescription{leftmargin=1.5cm,labelindent=\parindent}
\begin{description}
\addtolength{\itemsep}{-0.5\baselineskip}
\item[1.$\;$] Sort $A$ and $B$ in increasing order as $A(0)\ldots A(|A|-1)$ and $B(0)\ldots B(|B|-1)$.
\item[2.$\;$] For each $c\in C$,
\item [2.1.$\;\;$] Initialize $\lo \leftarrow 0$ and $\hi \leftarrow |B|-1$.
\item [2.2.$\;\;$] Repeat:
\item [2.2.1.$\;\;\;\;$] If $-c = A(\lo) + B(\hi)$, report witness ``$(A(\lo), B(\hi), c)$''
\item [2.2.2.$\;\;\;\;$] If $-c < A(\lo) + B(\hi)$ then decrement $\hi$, otherwise increment $\lo$.
\item [2.3.$\;\;\;$] Until $\lo = |A|$ or $\hi = -1$.
\item[3.$\;$] If no witnesses were found report ``no witness.''
\end{description}

Note that when a witness is discovered in Step 2.2.1 the algorithm continues to search for more witnesses involving $c$.
Since the elements in each row and each column of the $A+B$ matrix are distinct, it does not matter whether we increment 
$\lo$ or decrement $\hi$ after finding a witness.  We choose to increment $\lo$ in such situations; this choice is reflected in Lemma~\ref{lem:contour} 
and its applications in Sections~\ref{sect:rand-param} and \ref{sect:det-param}.

Define the {\em contour} of $x$, $\contour(x,A+B)$, to be the sequence of positions $(\lo,\hi)$ encountered while searching for $x$ in $A+B$. 
When $A+B$ is understood from context we will write it as $\contour(x)$.
If $A+B$ is viewed as a topographic map, with the lowest point in the NW corner and highest point in the SE corner, $\contour(x)$ represents the
path taken by an agent attempting to stay as close to altitude $x$ as possible, starting in the NE corner (at position $(0,|B|-1)$)
and ending when it falls off the western or southern side of the map.  Lemma~\ref{lem:contour} is straightforward.

\begin{lemma}\label{lem:contour}
Every occurrence of $x$ in $A+B$ lies on $\contour(x)$.
Let $y=(A+B)(i,j)$ be any element of $A+B$.
Then $y>x$ iff either $(i,j)$ lies strictly below $\contour(x)$
or both $(i,j)$ and $(i,j-1)$ lie on $\contour(x)$.
Similarly, $y\le x$ iff either $(i,j)$ lies strictly above $\contour(x)$
or both $(i,j)$ and $(i+1,j)$ lie on $\contour(x)$.
\end{lemma}

\section{A Subquadratic \ThreeSUM{} Decision Tree}\label{sect:nonuniform}

Recall that we are given a set $A\subset \mathbb{R}$ of reals and must determine
if there exist $a,b,c\in A$ summing to zero.  We first state the algorithm, then establish its correctness and efficiency.

\setdescription{leftmargin=1.5cm,labelindent=\parindent}
\begin{description}
\addtolength{\itemsep}{-0.25\baselineskip}
\item[1.$\;$] Sort $A$ in increasing order as $A(0)\ldots A(n-1)$.
Partition $A$ into $\ceil{n/g}$ groups $A_0,\ldots, A_{\ceil{n/g}-1}$  of size at most $g$, 
where $A_i \bydef \{A(ig), \ldots, A((i+1)g-1)\}$ and $A_{\ceil{n/g}-1}$ may be smaller.
The first and last elements of $A_i$ are $\min(A_i) = A(ig)$ and $\max(A_i) = A((i+1)g-1)$.

\item[2.$\;$] Sort $D \bydef \bigcup_{i\in [n/g]} \left(A_i-A_i\right) = \{a - a' \;|\; a,a' \in A_i \mbox{ for some } i\}$.

\item[3.$\;$] For all $i,j \in [n/g]$, sort $A_{i,j} \bydef A_i + A_j = \{a + b \;|\; a\in A_i \mbox{ and } b\in A_j\}$.

\item[4.$\;$] For $k$ from 1 to $n$,
\item [4.1.$\;\;$] Initialize $\lo \leftarrow 0$ and $\hi \leftarrow \floor{k/g}$ to be the group index of $A(k)$.
\item [4.2.$\;\;$] Repeat:
\item [4.2.1.$\;\;\;\;$] If $-A(k) \in A_{\lo,\hi}$, report ``solution found'' and halt.
\item [4.2.2.$\;\;\;\;$] If $\max(A_{\lo}) + \min(A_{\hi}) > -A(k)$ then decrement $\hi$, otherwise increment $\lo$.
\item [4.3.$\;\;\;$] Until $\hi < \lo$.

\item[5.$\;$] Report ``no solution'' and halt.
\end{description}

With appropriate modifications this algorithm also solves the three-set version of \ThreeSUM, where the input is $A,B,C\subset\mathbb{R}$.

\paragraph{Efficiency of the Algorithm.}
Step 1 requires $n\log n$ comparisons.
By Lemmas~\ref{lem:Fredman} and \ref{lem:Buck}, 
Step 2 requires $O(n\log n + |D|) = O(n\log n + gn)$ comparisons to sort $D$.
Using Fredman's trick, Step 3 requires no comparisons at all, given the sorted order on $D$.
(If $a,a' \in A_i$ and $b,b'\in A_j$, $a + b < a' + b'$ holds iff $a-a' < b'-b$.)
For each iteration of the outer loop (Step 4) there are at most $\ceil{n/g}$ iterations of the inner loop (Step 4.2) since each iteration
ends by either incrementing $\lo$ or decrementing $\hi$.
In Step 4.2.1 we can determine whether $-A(k)$ is in $A_{\lo,\hi}$ with a binary search, in $\log|A_{\lo,\hi}| = \log(g^2)$ comparisons.
In total the number of comparisons is on the order of $n\log n + gn + (n^2\log g)/g$, which is $O(n^{3/2}\sqrt{\log n})$ when $g = \sqrt{n\log n}$.

\paragraph{Correctness of the Algorithm.}
The purpose of the outer loop (Step 4) is to find $a,b\in A$, for which $a,b \le A(k)$ and $a + b + A(k) = 0$.
This is tantamount
to finding indices $\lo,\hi$ for which $a\in A_{\lo}, b\in A_{\hi}$, and $-A(k) \in A_{\lo,\hi}$.
We maintain the loop invariant that if there exist $a,b$ for which $a+b +A(k)=0$, then both of $a$ and $b$
lie in $A_{\lo}, A_{\lo+1}, \ldots, A_{\hi}$.  Suppose the algorithm has not halted in Step 4.2.1, that is, there are no solutions
with $a\in A_{\lo}, b\in A_{\hi}$.  If $\max(A_{\lo}) + \min(A_{\hi}) > -A(k)$ then there can clearly be no solutions with 
$b \in A_{\hi}$ since $b\geq \min(A_{\hi})$, so decrementing $\hi$ preserves the invariant.
Similarly, if $\max(A_{\lo}) + \min(A_{\hi}) < -A(k)$ then there can be no solutions with $a\in A_{\lo}$ since $a\leq \max(A_{\lo})$, 
so incrementing $\lo$ preserves the invariant.
If it is ever the case that $\hi < \lo$ then, by the invariant, no solutions exist.

\paragraph{Algorithmic Implementation.}
This \ThreeSUM{} algorithm can be implemented to run in $O(n^2\log n)$ time while performing only 
$O(n^{3/2}\log n)$ comparisons.
Using any optimal sorting algorithm
Steps 1--3 can be executed in $O(gn\log(gn) + (n/g)^2\cdot g^2 \log g) = O(n^2\log n)$ time while using $O(gn\log(gn))$ comparisons.
Now the boxes $\{A_{i,j}\}$ have been explicitly sorted, so the binary searches in Step 4.2.1 can be executed in $O(\log g)$ time per search.
The total running time is $O(n^2\log n)$ and the number of comparisons is 
now minimized when $g=\sqrt{n}$, for a total of $O(n^{3/2}\log n)$ comparisons.
We do not know of any polynomial time \ThreeSUM{} algorithm that performs $O(n^{3/2}\sqrt{\log n})$ comparisons.

\section{Some Subquadratic \ThreeSUM{} Algorithms}\label{sect:uniform}

In our \ThreeSUM{} decision tree, 
sorting $D$ (Step 2) is a comparison-efficient way to accomplish Step 3, but it only lets us {\em deduce}
the sorted order of the boxes $\{A_{i,j}\}$.
It does not give us a useful representation of these sorted orders, namely 
one that lets us implement each comparison of the binary search in Step 4.2.1 in $O(1)$ time.
In this section we present several methods for sorting the boxes 
based on bichromatic dominating pairs,
as in Chan~\cite{Chan08}
and Bremner et al.~\cite{BremnerCDEHILPT14}; see Lemma~\ref{lem:redblue}.
The total time spent performing binary searches in Step 4.2.1 will be $O(n^2\log g/g)$, so our goal is to make
$g$ as large as possible, provided that the cost of sorting the boxes is of a lesser order.

\paragraph{Overview.}
As a warmup we  
give, in Section~\ref{sect:simple-subquad-alg}, a relatively simple subquadratic \ThreeSUM{} algorithm running in 
$O\mathopen{}\left(n^2(\log\log n)^{3/2}/(\log n)^{1/2}\right)\mathclose{}$ time.  
In Section~\ref{sssect:faster3sum} we present a more sophisticated algorithm, some of whose parameters can be selected
either deterministically or randomly.
Sections~\ref{sect:rand-param} and \ref{sect:det-param} give two parameterizations of the algorithm, which
lead to an $O\mathopen{}\left(n^2 (\log\log n/\log n)^{2/3}\right)\mathclose{}$ time deterministic \ThreeSUM{} algorithm
and $O\mathopen{}\left(n^2 (\log\log n)^2/ \log n\right)\mathclose{}$-time randomized \ThreeSUM{} algorithm.

\subsection{A Simple Subquadratic \ThreeSUM{} Algorithm}\label{sect:simple-subquad-alg}

Choose the group size to be $g = \Theta(\sqrt{\log n/\log\log n})$.  The algorithm enumerates {\em every}
permutation $\pi : [g^2]\rightarrow [g]^2$, where $\pi = (\pi_{\row},\pi_{\col})$ is decomposed into
row and column functions $\pi_{\row}, \pi_{\col} : [g^2] \rightarrow [g]$.
By definition $\pi$ is the correct sorting permutation iff
$A_{i,j}(\pi(t)) < A_{i,j}(\pi(t+1))$ for all $t\in [g^2-1]$.\footnote{Without loss of generality we can assume $A_{i,j}$ is totally ordered.  See Lemma~\ref{lem:tiebreaking}.}
Since $A_{i,j} = A_i+A_j$ this inequality can 
also be written $A_i(\pi_{\row}(t)) + A_j(\pi_{\col}(t)) < A_i(\pi_{\row}(t+1)) + A_j(\pi_{\col}(t+1))$.
By Fredman's trick this is equivalent to saying that the (red) point $p_j$ dominates the (blue) point $q_i$, where
\begin{align*}
p_{j} &= \left(A_j(\pi_{\col}(1)) - A_j(\pi_{\col}(0)), \,\ldots\,, A_j(\pi_{\col}(g^2-1)) - A_j(\pi_{\col}(g^2-2))\right)\\
q_{i} &= \left(A_i(\pi_{\row}(0)) - A_i(\pi_{\row}(1)), \,\ldots\,, A_i(\pi_{\row}(g^2-2)) - A_i(\pi_{\row}(g^2-1))\right).
\end{align*}

We find all such dominating pairs.
By Lemma~\ref{lem:redblue} the time to report red/blue dominating pairs, over all $(g^2)!$ invocations of the procedure, 
is $O\mathopen{}\left((g^2)! c_\epsilon^{g^2-1} (2n/g)^{1+\epsilon} + (n/g)^2\right)\mathclose{}$, the last term being the total size of the outputs.
For $\epsilon=1/2$ and $g = \frac{1}{2}\sqrt{\log n/\log\log n}$
the first term is negligible.  The total running time is therefore 
$O((n/g)^2)$ for dominance reporting
and $O(n^2\log g/g) = O\mathopen{}\left(n^2 (\log\log n)^{3/2}/(\log n)^{1/2}\right)\mathclose{}$
for the binary searches in Step 4.2.1.

Since there are at most $g^{8g}$ realizable permutations of $A_{i,j}$, not $(g^2)!$ (see Lemma~\ref{lem:sortX+Y} and Fredman~\cite{Fredman76}),
we could possibly shave off another $\sqrt{\log\log n}$ factor by setting $g = \Theta(\sqrt{\log n})$.  
However, with a bit more work it is possible to save $\poly(\log n)$ factors, as we now show.

\subsection{A Faster \ThreeSUM{} Algorithm}\label{sssect:faster3sum}

To improve the running time of the simple algorithm we must sort larger boxes.
Our approach is to partition the blocks into layers and sort each layer separately.
So long as each layer has size $\Theta(\log n)$, the cost of red/blue dominance reporting will
be negligible.  The main difficulty is that the natural boundaries between layers are unknown and different for each of the blocks in $\{A_{i,j}\}$.

Let $P\subset [g]^2$ be a set of $p$ positions in the $g\times g$ grid that includes positions $(0,0)$ and $(g-1,g-1)$.
How we select the remaining $p-2$ positions in $P$ will be addressed later.
For each $(l,m)\in P$, consider $\contour(A_{i,j}(l,m),A_{i,j})$, that is,
the path in $[g]^2$ taken by the standard \ThreeSUM{} algorithm of Section~\ref{sect:quadratic3SUM}
when searching for $A_{i,j}(l,m)$ inside $A_{i,j}$.
Clearly $\contour(A_{i,j}(l,m))$ goes through position $(l,m)$.  
For any two $(l,m),(l',m')\in P$, $\contour(A_{i,j}(l,m))$ and $\contour(A_{i,j}(l',m'))$ may intersect in several places (see Figure~\ref{fig:searchpath})
though they never cross.  According to Lemma~\ref{lem:contour}, the two contours define a {\em tripartition} $(R,S,T)$ 
of the positions of $[g]^2$ into three regions,
 where
\begin{align*}
A_{i,j}(R) &\subset (-\infty,\, A_{i,j}(l,m)]\\
A_{i,j}(S) &\subset (A_{i,j}(l,m),\, A_{i,j}(l',m'))\\
A_{i,j}(T) &\subset [A_{i,j}(l',m'),\, \infty)
\end{align*}
Here $A_{i,j}(X) = \{A_{i,j}(x) \:|\: x\in X\}$ for a subset $X\subseteq [g]^2$.
Note that $(R,S,T)$ is fully determined by the shapes of the contours, not the specific contents of 
$A_{i,j}$.\footnote{We continue to assume that ties are broken to make $A_{i,j}$ totally ordered.  Refer to Lemma~\ref{lem:tiebreaking}.}
See Figure~\ref{fig:tripartition}(a) for an illustration.

\newcommand{\cellred}{\cellcolor[rgb]{1,0,0}}
\newcommand{\cellblue}{\cellcolor[rgb]{0,.5,1}}
\newcommand{\cellgreen}{\cellcolor[rgb]{0,1,0}}
\newcommand{\cellyellow}{\cellcolor[rgb]{1,1,0}}
\begin{figure}[tb]
\centering
\begin{tabular}{|c|c|c|c|c|c|c|c|c|c|}
\hline
 250 &  272 &  362 &  368 &  372 &  385 & \cellblue 416 & \cellblue 546 & \cellgreen 549 & \cellgreen 606 \\
\hline
 289 &  311 &  401 &  407 &  411 & \cellblue 424 & \cellgreen 455 & \cellyellow 585 & \cellyellow 588 &  645 \\
\hline
 299 &  321 &  411 &  417 &  421 & \cellblue 434 & \cellyellow 465 &  595 &  598 &  655 \\
\hline
 311 &  333 &  423 &  429 &  \cellblue 433 & \cellblue \emph{446} & \cellyellow 477 &  607 &  610 &  667 \\
\hline
 325 &  347 &  437 & \cellblue 443 & \cellblue 447 &  460 & \cellyellow 491 &  621 &  624 &  681 \\
\hline
 331 &  353 & \cellblue 443 & \cellblue 449 &  453 &  466 & \cellyellow 497 &  627 &  630 &  687 \\
\hline
 363 & \cellblue 385 & \cellblue 475 &  481 &  485 &  498 & \cellyellow 529 &  659 &  662 &  719 \\
\hline
 384 & \cellblue 406 &  496 &  502 &  506 &  519 & \cellyellow 550 &  680 &  683 &  740 \\
\hline
 412 & \cellblue 434 &  524 &  530 &  534 &  \cellyellow 547 & \cellyellow \emph{578} &  708 &  711 &  768 \\
\hline
 415 & \cellblue 437 &  527 &  533 &  537 & \cellyellow 550 &  581 &  711 &  714 &  771 \\
\hline
\end{tabular}
\caption{\label{fig:searchpath}A subset of a block defined by two search paths. The red and
  green are the two search paths for elements at position (3,5) and
  (8,6) respectively (yellow are shared points). The subset defined
  is the elements on the paths and the elements between them.}
\end{figure}

A contour is defined by at most $2g-1$ comparisons between the
search element and elements of the block.  Suppose
that $\tau = (\tau_{\row},\tau_{\col})$ is purported to be  
$\contour(A_{i,j}(l,m))$, that is, $\tau(0) = (\tau_{\row}(0),\tau_{\col}(0)) = (0,g-1)$ is the starting position of $(\lo,\hi)$
and $\tau(t+1) \in \tau(t) + \{(1,0), (0,-1)\}$ depending on whether $\lo$ is incremented or $\hi$ is decremented after 
the $(t+1)$th comparison.  The contour ends at the first $t^\star$ for which 
$\tau(t^\star) = (g,\cdot)$ or $(\cdot, -1)$ depending on whether the search for $A_{i,j}(l,m)$ falls off the southern or western boundary of $A_{i,j}$.
Clearly $\tau$ is the correct contour if and only if
\begin{align*}
A_{i,j}(l,m) &< A_{i,j}(\tau(t))		& \mbox{ when $\tau(t+1) = \tau(t) + (0,-1)$}\\
A_{i,j}(l,m) &> A_{i,j}(\tau(t))		& \mbox{ when $\tau(t+1) = \tau(t) + (1,0)$}
\end{align*}
for every $t\in [t^\star]$, excluding the $t$ for which $\tau(t) = (l,m)$ since in this case we have equality: $A_{i,j}(l,m) = A_{i,j}(l,m)$.
Restating this, $\tau$ is the correct contour if the (red) point $p_j$ dominates the (blue) point $q_i$, defined as
\begin{align*}
p_j &= \left(\ldots, \;\sigma(t)\left(A_j(\tau_{\col}(t)) - A_j(m)\right), \;\ldots \right)\\
q_i &= \left(\ldots, \;\sigma(t)\left(A_i(l) - A_i(\tau_{\row}(t))\right),  \;\ldots \right),
\end{align*}
where $\sigma(t) \in \{1,-1\}$ is the proper sign:
\[
\sigma(t) = \left\{
\begin{array}{ll}
1 & \mbox{ when $\tau(t+1)=\tau(t) + (0,-1)$, and}\\
-1 & \mbox{ when $\tau(t+1)=\tau(t) + (1,0)$.}
\end{array}
\right.
\]
The coordinate $t$ for which $\tau(t)=(l,m)$ is, of course, omitted from $p_j$ and $q_i$, so 
both vectors have length at most $2g-2$.

Call a pair $(\tau,\tau')$ of contours {\em legal}
if
\begin{enumerate}
\item[(i)] Whenever $\tau$ and $\tau'$ do not intersect, $\tau$ is above $\tau'$.
\item[(ii)] There are two $(l,m),(l',m')\in P$ such that $\tau$ contains  $(l,m)$ and $\tau'$ contains $(l',m')$.
\item[(iii)] Let $(R,S,T)$ be the tripartition of $[g]^2$ defined by $(\tau,\tau')$, 
where $S$ are those positions lying strictly between $(l,m)$ and $(l',m')$. 
Then $P\cap S = \emptyset$ and $|S|\le s$, where $s$ is a parameter to be determined.
\end{enumerate}
Let us clarify criterion (iii).  It states that if $A_{i,j}$ is any {\em specific} box for which
$(\tau,\tau')$ are correct contours of $A_{i,j}(l,m)$ and $A_{i,j}(l',m')$, the number of positions $(l'',m'')\in [g]^2$
for which $A_{i,j}(l'',m'') \in (A_{i,j}(l,m), A_{i,j}(l',m'))$ is at most $s$, and no such position appears in $P$.

Our algorithm enumerates every legal pair $(\tau,\tau')$ of contours, at most $2^{4g}$ in total.
Let $(l,m),(l',m')\in P$ be the points lying on $\tau,\tau'$ and $(R,S,T)$ be the tripartition of $[g]^2$ defined by $(\tau,\tau')$.
For each $(\tau,\tau')$ the algorithm enumerates every realizable permutation $\pi : [|S|] \rightarrow S$ of
the elements at positions in $S$.  
By Lemma~\ref{lem:sortX+Y} there are $O\mathopen{}\left({s\choose 2}^{2g}\right)\mathclose{} < 2^{4g \log s}$ such permutations, 
which can be enumerated in $O(2^{4g\log s})$ time.
For each $(\tau,\tau',\pi)$ we create red points $\{p_j\}_{j\in [n/g]}$ and blue points $\{q_i\}_{i\in [n/g]}$ in $\mathbb{R}^{4g+s-5}$ 
such that $p_j$ dominates $q_i$ iff 
$\tau=\contour(A_{i,j}(l,m))$ and $\tau'=\contour(A_{i,j}(l',m'))$ are the correct contours (w.r.t.~$A_{i,j}$)
and $\pi$ is the correct sorting permutation of $A_{i,j}(S)$.
The first $4g-4$ coordinates encode the correctness of $\tau$ and $\tau'$ and the last $s-1$ coordinates
encode the correctness of $\pi$.

According to Lemma~\ref{lem:redblue}, the time to report all dominating pairs 
is $O(p(n/g)^2 + 2^{4g}\cdot 2^{4g\log s}\cdot (c_{\epsilon})^{4g+s-5}(2n/g)^{1+\epsilon})$.
The first term is the output size, since by
criterion (iii) of the definition of {\em legal}, at most $p-1$ pairs are reported for each of the $(\ceil{n/g})^2$ boxes.
There are $2^{4g}2^{4g\log s}$ choices for $(\tau,\tau',\pi)$ and the dimension of the point set is at most $4g+s-5$, 
but could be smaller if the contours happen to be short or $|S| < s$.  
Fixing $\epsilon=1/2$, if $g\log s$ and $g+s$ are both $O(\log n)$ (with a sufficiently small leading constant)
the running time of the algorithm will be dominated by the time spent reporting the output.

Call a box $A_{i,j}$ {\em bad} if the output of the dominating pairs algorithm fails to determine its sorted order.
The only way a box can be bad is if an otherwise legal $(\tau,\tau')$ with tripartition $(R,S,T)$ were correct for
$A_{i,j}$  but failed to be legal because $|S| > s$, leaving the sorted order of $A_{i,j}(S)$ unknown.

If all boxes are not bad we can search for $x\in \mathbb{R}$ in $A_{i,j}$ in $O(\log g)$ time using binary search, as follows.
Each box $A_{i,j}$ was associated with a list of $p-1$ triples of the form $(\tau,\tau',\pi)$ returned by the dominating pairs algorithm, 
one for each pair of successive elements in $A_{i,j}(P)$.
The first step is to find the predecessor of $x$ in $A_{i,j}(P)$, that is, to find the consecutive $(l,m),(l',m')\in P$ 
for which $A_{i,j}(l,m) \le x < A_{i,j}(l',m')$.  Let $(\tau,\tau',\pi)$ be the triple associated with $(l,m),(l',m')$
and $(R,S,T)$ be the tripartition of $(\tau,\tau')$.
Each legal, realizable $(\tau,\tau',\pi)$ is encoded as a bit string with length $4g(1+\log s)$, which must
fit comfortably in one machine word.  Before executing the algorithm proper we build, in $o(n)$ time, a lookup table
indexed by tuples $(\tau,\tau',\pi,r)$ that contains the location in $[g]^2$ of the element with rank $r$ in $S$,
sorted according to $\pi$.  Using this lookup table it is 
straightforward to perform a binary search for $x$ in $A_{i,j}(S)$, in $O(\log|S|) = O(\log g)$ time.

\begin{figure}
\centering
\begin{tabular}{c@{\hcm[1]}c}
\scalebox{.4}{\includegraphics{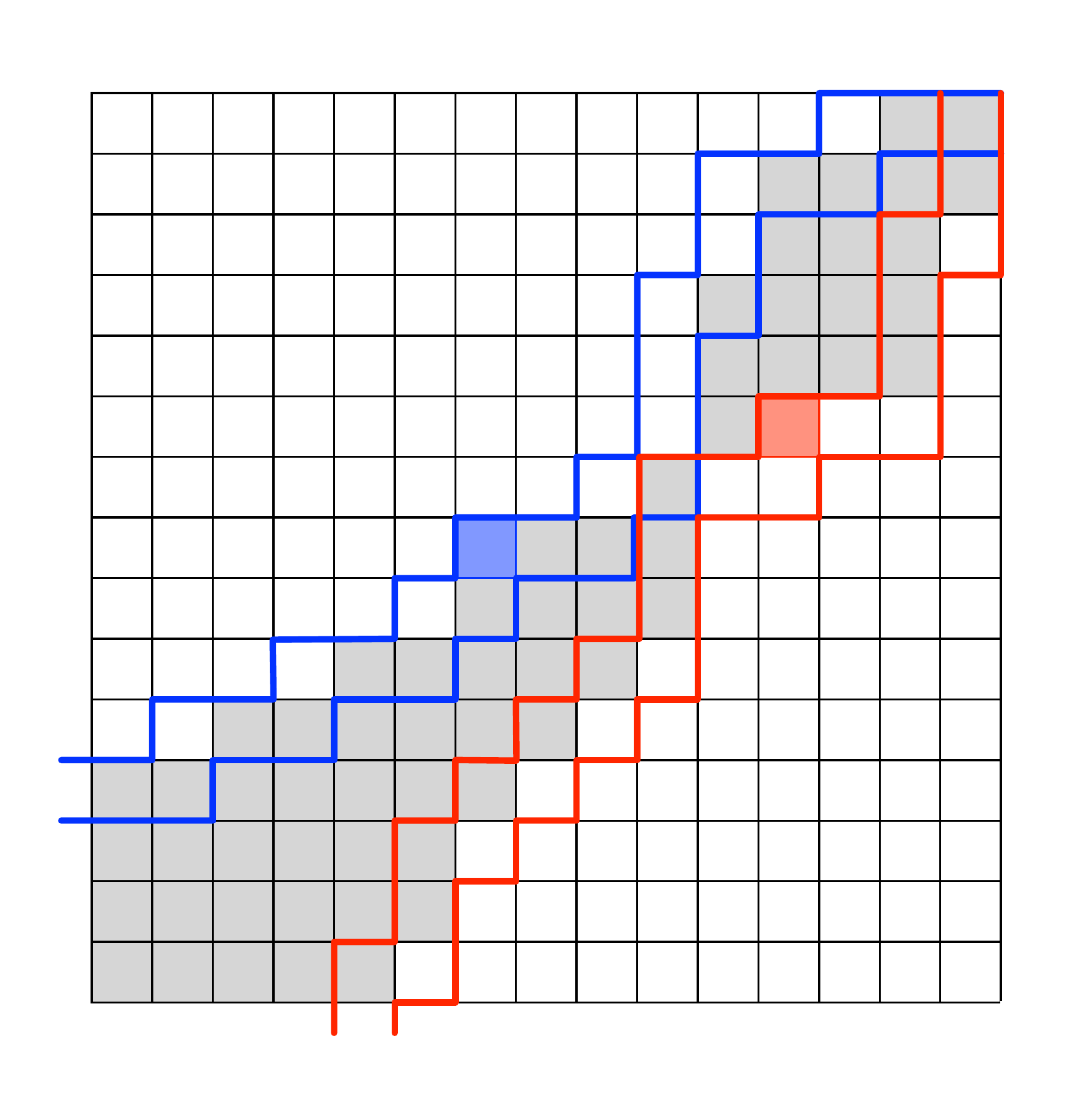}}
&
\scalebox{.4}{\includegraphics{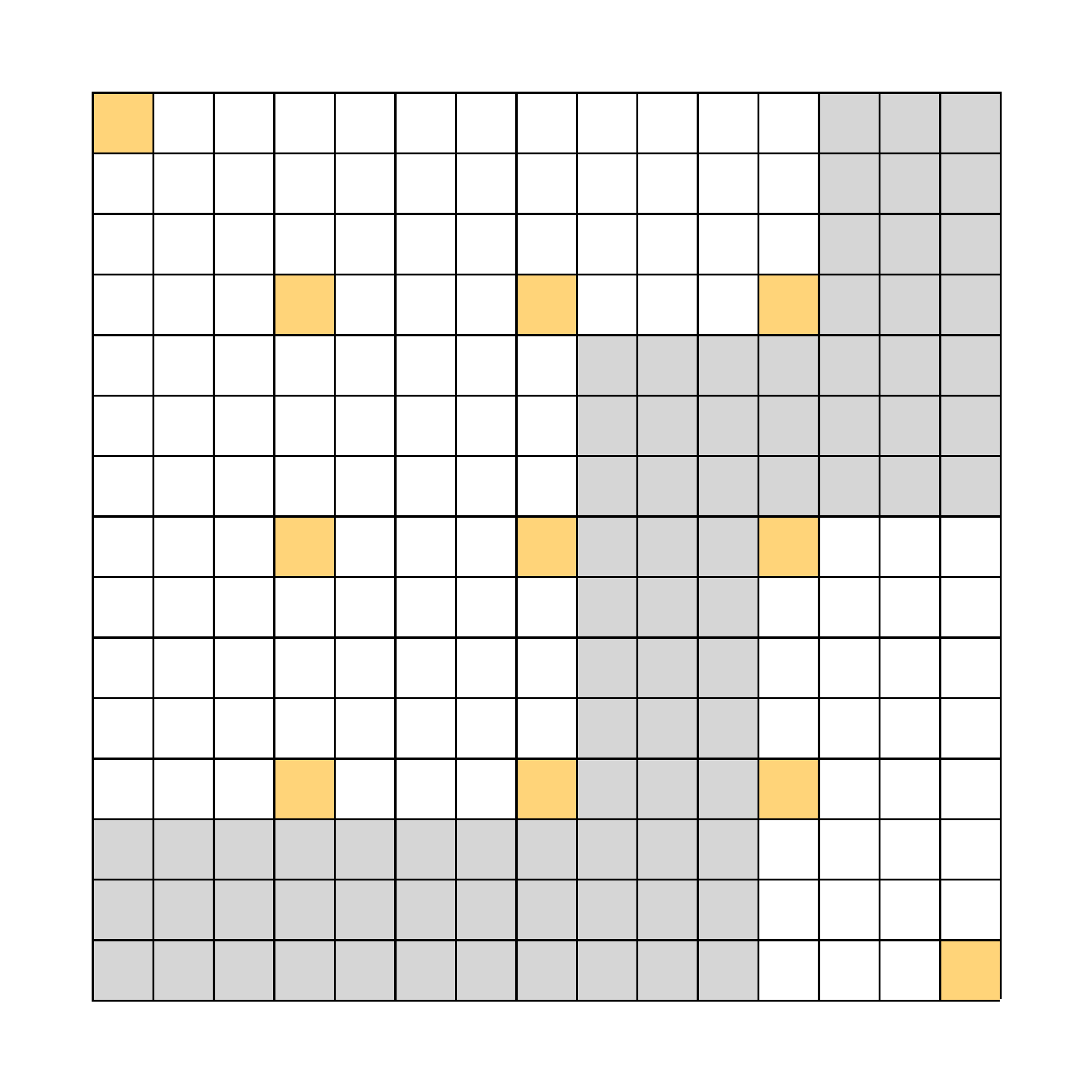}}\\
(a) & (b)
\end{tabular}
\caption{\label{fig:tripartition}Illustrations of tripartitions defined by two contours in a $[15]\times [15]$ grid.
Left: the blue and red locations are in $P$.  Two possible contours are indicated by blue and red paths.
They define a tripartition $(R,S,T)$ with $S$ marked in gray.
Right: $P$ is chosen to include two corner locations and an evenly spaced $q\times q$ grid.
Any tripartition $(R,S,T)$ defined by a legal pair of contours has $P\cap S = \emptyset$, implying
that $|S| \le 2g^2/(q+1)$.  An example of an $S$ nearly achieving that size is marked in gray.}
\end{figure}

\subsection{A Randomized Parameterization of the Algorithm}\label{sect:rand-param}

Throughout let $\delta>0$ be a sufficiently small constant.
In the randomized implementation of our algorithm we choose
$g = s = \delta\ln n/\ln\ln n$ and $p = 2+3\delta\ln n$.  The points $(0,0)$ and $(g-1,g-1)$ must be in $P$
and the remaining $3\delta\ln n$ points are chosen uniformly at random.  With these
parameters the probability of a box being {\em bad} is sufficiently low to keep the expected
cost per search $O(\log g)$.

\begin{lemma}\label{lem:badprob}
The probability a particular box is bad is at most $1/g$.
\end{lemma}

\begin{proof}
Let $\pi$ be the sorted order for some box $A_{i,j}$.
The probability $A_{i,j}$ is bad is precisely the probability that there are
$s+1$ consecutive elements (according to $\pi$) that are not included in $P$.
The probability that this occurs for a particular set of $s+1$ elements is
less than $(1-(p-2)/g^2)^{s+1} < e^{-(p-2)(s+1)/g^2} < e^{-3\ln\ln n} < 1/g^3$.
By a union bound over all $g^2-(s+1)$ sets of $s+1$ consecutive elements,
the probability $A_{i,j}$ is bad is at most $1/g$.
\end{proof}

The expected time per search is therefore $O(\log g) + 1/g\cdot O(g) = O(\log g)$. 
By linearity of expectation the expected total running time is 
$O(p(n/g)^2 + n^2 (\log g)/g) = O(n^2 (\log\log n)^2/\log n)$.

\begin{remark}
We could have set the parameters differently and achieved the same running time.
For example, setting $g = p = \Theta(\log n/\log\log n)$ and $s=\Theta(\log n)$ would also work.
The advantage of keeping $s=O(\log n/\log\log n)$ is simplicity:
we can afford to enumerate all $s!$ permutations of $S\subset [g]^2$ rather than explicitly 
construct a hyperplane arrangement in order to enumerate
only those {\em realizable} permutations of $S$.
\end{remark}

\paragraph{High Probability Bounds.}
The running time of the algorithm may deviate from its expectation with non-negligible probability
since the {\em badness} events for the boxes $\{A_{i,j}\}$ can be strongly positively correlated.
The easiest way to obtain high probability bounds is simply to choose $L=c\log n$
random point sets $\{P_l\}_{l\in [L]}$, estimate the cost of the algorithm under each 
point set, then execute the algorithm under the point set with the best estimated cost.
The first step is to run a truncated version of the algorithm in order to determine which queries will be asked in Step 4.2.1.
Rather than answer the query $-A(k) \in A_{i,j}$ we simply record the triple $(k,i,j)$ in a list $\mathcal{Q}$ to be answered later.
The running time of the algorithm under $P_l$ 
is $O(n^2(\log g)/g)$ plus $g$ times the number of bad triples in $\mathcal{Q}$, that is, 
those $(k,i,j)$ for which $A_{i,j}$ is bad according to $P_l$.

Let $\epsilon_l$ be the true fraction of bad triples in $\mathcal{Q}$ according to $P_l$ and
$\hat\epsilon_l$ be the estimate of $\epsilon_l$ obtained by the following 
procedure.  Sample $M = cg^2\ln n$ elements of $\mathcal{Q}$ uniformly at random
and for each, test whether the given block is bad according to $P_l$ by sorting its elements, in $O(g^2\log g)$ time.
If $X$ is the number of blocks discovered to be bad, report the estimate $\hat\epsilon_l = X/M$.
By a standard version of the Chernoff bound\footnote{If $X$ is the number of successes in $n$ independent Bernoulli
trials, $\Pr(X>\E(X) + t)$ and $\Pr(X<\E(X)+t)$ are both upper bounded by $e^{-2t^2/n}$.  
See~\cite[Thm.~1.1]{DubhashiPanconesi09}.}   we have
\[
\Pr(|\hat\epsilon_l - \epsilon_l| > 1/g)
= \Pr(|X - \E(X)| > M/g) < 2e^{-2(M/g)^2/M} = 2n^{-2c}.
\]
By Lemma~\ref{lem:badprob}, $\E(\epsilon_l) = 1/g$ for each $l$,
so by Markov's inequality $\Pr(\min_{l\in L}\{\epsilon_l\} \le 2/g) \ge 1 - 2^{-L} = 1-n^{-c}$.
With high probability, each $\hat\epsilon_l$ deviates from $\epsilon_l$ by at most $1/g$,
so the running time of the algorithm will be within a constant factor of its expectation with 
probability $1-O(n^{-c})$.  The time to pick the best point set $P_{l^\star}$, $l^\star$ being $\argmin_{l\in[L]} \{\hat\epsilon_l\}$,
is $O(LMg^2\log g) = o(\log^6 n)$.
We could set $L$ and $M$ as high as $n/\polylog(n)$, making the probability that 
the algorithm deviates from its expectation exponentially small, $\exp(-n/\polylog(n))$.

\subsection{A Deterministic Parameterization of the Algorithm}\label{sect:det-param}

We achieve a subquadratic worst-case \ThreeSUM{} algorithm by choosing $g,s,p,$ and $P$ such that no block can be bad.
Fix $g = (\delta\log n)^{2/3}(\log\log n)^{1/3}$ and $p = 2+ q^2 < (\delta\log n)^{2/3}(\log\log n)^{4/3}$ for an integer $q$ to be determined.
Aside from the two obligatory points, $P$ contains an evenly spaced $q\times q$ grid in $[g]\times [g]$.
Setting $\Delta = \ceil{\frac{g+1}{q+1}}$, $P$ is defined as
\[
P = \{(0,0),(g-1,g-1)\} \cup \{(k\Delta-1, l\Delta-1) \;|\; \mbox{ where $1\le k,l\le q$}\}.
\]
See Figure~\ref{fig:tripartition}(b).

We now argue that no box can be bad if $s = 2g(\Delta-1)$.
For any legal pair of contours $(\tau,\tau')$, in the corresponding tripartition $(R,S,T)$
no element of $P$ can be contained in $S$, that is, in any row (or column) containing elements of $P$, the width (or height)
of the band $S$ is at most $\Delta-1$.  Since both $\tau$ and $\tau'$ are monotone paths in $[g]^2$
(non-decreasing by row, non-increasing by column),
we always have $|S| < 2g(\Delta-1) < 2g^2/(q+1) < 2\delta\log n$.   See Figure~\ref{fig:tripartition}(b) for a worst-case example.
For $\delta$ sufficiently small the overhead for reporting dominating pairs will be negligible.
The overall running time is therefore 
$O(n^2(\log g)/g) = O(n^2 / (\log n/\log\log n)^{2/3})$.

\section{Linear Degeneracy Testing}\label{sect:lineardegeneracy}

Recall that we are given a set $S\subset \mathbb{R}$ and a function $\phi(x_1,\ldots,x_k) = \alpha_0 + \sum_{i=1}^k \alpha_i x_i$, for some real coefficients
$\{\alpha_i\}$.  The problem is to determine if there is a point $(x_1,\ldots,x_k) \in S^k$ where $\phi$ is zero.
As we show below, our \ThreeSUM{} decision tree can be generalized in a straightforward way to solve
$k$-\LDT{} with $O(n^{k/2}\sqrt{\log n})$ comparisons, when $k\ge 3$ is odd.
Unfortunately, we do not see how to generalize our \ThreeSUM{} {\em algorithms} to solve
$k$-\LDT{} in $O(n^{(k+1)/2}/\polylog(n))$ time, for any odd $k\ge 5$.

\begin{proof} (of Theorem~\ref{thm:kLDT})
Define $\alpha \cdot S$ to be the set $\{\alpha\cdot a \;|\; a\in S\}$, where $\alpha\in\mathbb{R}$.
Begin by sorting the sets
\begin{align*}
A &= \{\alpha_0 + a_1 + a_2 + \cdots + a_{(k-1)/2} \;|\; \mbox{$a_i \in \alpha_i \cdot S$, for each $i>0$}\}
\intertext{and}
B &= \{a_{(k+1)/2} + \cdots + a_{k-1} \;|\; \mbox{$a_i \in \alpha_i \cdot S$, for each $i$}\}
\end{align*}

We have effectively reduced $k$-\LDT{} to an unbalanced \ThreeSUM{} problem.  Letting
$C$ be the set $\alpha_k \cdot S$, the problem is to determine if there exist $a\in A, b\in B, c\in C$ such that $a+b+c=0$.
Note that $|A|=|B|=n^{(k-1)/2}$ whereas $|C|=n$.  The standard \ThreeSUM{} algorithm of Section~\ref{sect:quadratic3SUM}
performs $|C| \cdot (|A| + |B|) = n^{(k+1)/2}$ comparisons.
Generalizing the decision tree of Section~\ref{sect:nonuniform} (from one list to three)
we can solve unbalanced \ThreeSUM{} using $O(g(|A| + |B|) + g^{-1} |C| (|A|+|B|)\log g)$ comparisons, which is $O(n^{k/2}\sqrt{\log n})$ when $g=\sqrt{n\log n}$.
\end{proof}

Our subquadratic \ThreeSUM{} algorithms do not extend naturally to unbalanced instances.  When $g=\polylog(n)$
we can no longer afford to {\em explicitly} sort all $g\times g$ boxes in $A+B$ as this would require at least 
$\Omega(|A|\cdot |B|/g^2) = \Omega(n^{k-1}/\polylog(n))$ time.\footnote{Note that there are only $O(|C|(|A|+|B|)/g)$ boxes of interest.  The dominating pairs approach does not let us sort an arbitrary selection of boxes in constant time per box,
but it is possible to accomplish this task in a more powerful model of computation.  On a souped-up word RAM with $O(g^2\log n)$-bit words
and a couple non-standard unit-time operations, any $g\times g$ box can be sorted in $O(1)$ time.
Simulating such a unit-time operation on the traditional word RAM is a challenging problem.}

\section{Zero Triangles}\label{sect:ZeroTriangle}

We consider a matrix product called {\em target-min-plus} that subsumes the
$(\min,+)$-product (aka distance product) and the \ZeroTriangle{} problem of~\cite{WilliamsW13}.
Given real matrices $A\in (\mathbb{R}\cup\{\infty\})^{r \times s}, B\in (\mathbb{R}\cup\{\infty\})^{s \times t},$ 
and a target matrix $T\in (\mathbb{R}\cup\{-\infty,\infty\})^{r\times t}$,
the goal is to compute $C=\targetplus(A,B,T)$, where
\[
C(i,j) = \min\left\{\, A(i,k) + B(k,j) \:|\: k\in[s] \mbox{ and } A(i,k)+B(k,j) \ge T(i,j)\,\right\}
\]
as well as the matrix of witnesses, that is, the $k$ (if any) for which $C(i,j) = A(i,k) + B(k,j)$.  This operation reverts
to the $(\min,+)$-product when $T(i,j) = -\infty$.
It can also solve \ZeroTriangle{} on a weighted graph $G = (V, E, w)$ by setting $A,B,$ and $T$ as follows.
Let $A(i,j)=B(i,j) = w(i,j)$, where $w(i,j)\bydef \infty$ if $(i,j)\not\in E$,
and let $T(i,j) = -w(i,j)$ if $(i,j)\in E$ and $\infty$ otherwise.
If $C(i,j) = T(i,j)$ then there is a zero weight triangle containing $(i,j)$ and the witness matrix gives the third corner of the triangle.

The trivial target-min-plus algorithm runs in $O(rst)$ time and performs the same number of comparisons.
We can compute the target-min-plus product using fewer comparisons using Fredman's trick.

\begin{theorem}\label{thm:target-min-plus}
The decision-tree complexity of the target-min-plus product of three $n\times n$ matrices is
$O(n^{5/2}\sqrt{\log n})$.  This product can be computed in $O(n^3 (\log\log n)^2/\log n)$ time.
\end{theorem}

\begin{proof}
We first show that the 
target-min-plus product $\targetplus(A,B,T)$ can be determined with $O((r+t)s^2 + rt\log s)$ comparisons,
where $A,B,$ and $T$ are $r\times s, s\times t,$ and $r\times t$ matrices, respectively.
Begin by sorting the set
\[
D = \{A(i,k) - A(i,k'),\: B(k',j) - B(k,j) \;|\; i\in[r], j\in[t], \mbox{ and } k,k'\in [s]\}.
\]
By Lemma~\ref{lem:sortX+Y} the number of comparisons required to sort $D$ is $O(|D| + (r+t)s\log(rst)) = O((r+t)s^2 + (r+t)s\log(rst))$.
We can now deduce the sorted order on 
\[
S(i,j) = \{A(i,k) + B(k,j) \;|\; k\in[s]\},
\]
for any pair $(i,j)\in [r]\times [t]$, and can therefore find $C(i,j) = \min(S(i,j) \cap [T(i,j),\infty))$
with a binary search over $S(i,j)$ using $\log s$ additional comparisons.
The total number of comparisons is $O((r+t)s^2 + rt\log s)$.
Note that this provides no improvement when $A$ and $B$ are square, that is, when $r=s=t=n$.
Following Fredman~\cite{F76} we partition $A$ and $B$ into 
rectangular matrices and compute their target-min-plus products separately.

Choose a parameter $g$ and 
partition $A$ into $A_{0},\ldots,A_{\ceil{n/g}-1}$ and $B$ into $B_0,\ldots,B_{\ceil{n/g}-1}$
where $A_\ell$ contains columns $\ell g,\ldots,(\ell+1)g-1$ of $A$ and $B_\ell$ contains the 
corresponding rows of $B$.
For each $\ell\in[n/g]$, compute the target-min-plus product $C_\ell = \targetplus(A_\ell,B_\ell,T)$
and set $C(i,j) = \min_{\ell\in[n/g]} (C_{\ell}(i,j))$.
This algorithm performs $O((n/g)\cdot (ng^2 + n^2\log n))$ comparisons to compute $\{C_{\ell}\}_{\ell\in[n/g]}$
and $n^2(n/g)$ comparisons to compute $C$.  When $g=\sqrt{n\log n}$ the number of comparisons is $O(n^{5/2}\sqrt{\log n})$.

To compute the product efficiently we use the geometric dominance approach of 
Chan~\cite{Chan08} and Bremner et al.~\cite{BremnerCDEHILPT14}.
Choose a parameter $g=\Theta(\log n/\log\log n)$ and partition $A$ into $n\times g$ matrices $\{A_\ell\}$ and $B$ into $g\times n$ matrices $\{B_\ell\}$.
For each  $\ell \in [n/g]$ and permutation $\pi : [g] \rightarrow [g]$ we will find those pairs $(i,j) \in [n]^2$ 
for which $\pi$ is the sorted order on $\{A_{\ell}(i,k) + B_{\ell}(k,j) \;|\; k\in [g]\}$.\footnote{Break ties in any consistent fashion so that the sorted order is unique.}
Such a triple satisfies the inequality $A_{\ell}(i,\pi(k)) + B_{\ell}(\pi(k),j) < A_{\ell}(i,\pi(k+1)) + B_{\ell}(\pi(k+1),j)$, for all $k\in [g-1]$.
By Fredman's trick this is equivalent to saying that the (red) point
\begin{align*}
&\left(\ldots, A_{\ell}(i,\pi(k+1)) - A_{\ell}(i,\pi(k)),\ldots\right)
\intertext{dominates the (blue) point}
&\left(\ldots, B_{\ell}(\pi(k),j) - B_{\ell}(\pi(k+1),j),\ldots \right)
\end{align*}
in each of the $g-1$ coordinates.  By Lemma~\ref{lem:redblue} the total time for all $\ceil{n/g}\cdot g!$ invocations of the dominance
algorithm is $O((n/g)\cdot g! \cdot c_\epsilon^{g-1} (2n/g)^{1+\epsilon})$ plus the output size, which is precisely $n^2\ceil{n/g}$.  For $\epsilon = 1/2$
and $g=\Theta(\log n/\log\log n)$ the running time is $O(n^3/g)$.
We can now compute the target-min-plus product $C_\ell = \targetplus(A_\ell,B_\ell,T)$ in $O(n^2\log g)$ time by iterating over
all $(i,j)\in[n]^2$ and performing a binary search to find the minimum element in $\{A_{\ell}(i,k) + B_{\ell}(k,j) \;|\; k\in [g]\} \cap [T(i,j),\infty)$.
Since $C = \targetplus(A,B,T)$ contains the pointwise minima of $\{C_\ell\}$, the total time to compute the target-min-plus product 
is $O(n^3(\log g) / g) = O(n^3(\log\log n)^2/\log n)$.
\end{proof}

The $\sqrt{\log n}$ factor in the decision tree complexity of target-min-plus arises comes from the 
binary searches, $n/g$ searches per pair $(i,j)\in[n]^2$.  If the searches were sufficiently correlated (either 
for fixed $(i,j)$ or fixed $\ell$) then there would be some hope that we could evade the information theoretic lower bound of
$\Omega(\log g)$ per search.
Using random sampling we form a hierarchy of rectangular target-min-plus products such that the solutions at one level gives a hint
for the solutions at the next lower level.  The cost of finding the solution, given the hint from the previous level, is $O(1)$ in expectation.
The same approach lets us shave off another $\log\log n$ factor off the algorithmic complexity of target-min-plus.

\begin{theorem}\label{thm:target-min-plus-rand}
The randomized decision tree complexity of the target-min-plus product of three $n\times n$ matrices
is $O(n^{5/2})$.  It can be computed in $O(n^3 \log\log n/\log n)$ time with high probability.
\end{theorem}

\begin{proof}
As usual let $A,B,$ and $T$ be $n\times n$ matrices and $g$ be a parameter.  We will eventually set $g=\ceil{\sqrt{n}}$.
We partition the indices $[n]$ at $\log\log n$ levels.  Define 
$I_{l,p} = [p g 2^l, (p+1) g2^l)$ to be the $p$th interval at level $l$.
In other words, level-$l$ intervals have width $g2^l$ and a level-$(l+1)$ interval is the union of two level-$l$ intervals.
Form a series of nested index sets $[n] = J_0 \supset J_1 \supset \cdots \supset J_{\log\log n-1}$, 
such that $J_{l} \cap I_{l,p}$ is a uniformly random subset of $I_{l,p}$ of size $g$.
In other words, each element of $J_{l-1}$ is promoted to $J_{l}$ with probability 1/2, but in such a way
that $|J_{l} \cap I_{l,p}|$ is precisely its expectation $g$.

After generating the sets $\{J_l\}$ the algorithm sorts $D$ with
$O(n^2\log n + |D|) = O(n^{5/2})$ comparisons (see Lemma~\ref{lem:sortX+Y}), where
\[
D = \left\{
A(i,k) - A(i,k'), \;\; B(k',i) - B(k,i) \;\;\left|\;\; 
\begin{array}{l}
\mbox{$i\in[n]$ and $k,k' \in J_l \cap I_{l,p}$,}\\
\mbox{for some level $l$ and index $p$}
\end{array}
\right.
\right\}
\]
Fix $i,j\in[n]$.  We proceed to compute $C(i,j)$ with $O(n/g)$ comparisons with high probability.
If $K\subset [n]$ is a set of indices, define $\kappa(K)$ to be the witness of the target-min-plus
product restricted to $K$, that is,
\[
\kappa(K) = \argmin_{\substack{k\in K\: \text{such that\istrut[1]{2.5}}\\ A(i,k) + B(k,j) \ge T(i,j)}} (A(i,k) + B(k,j)).
\]
There may, in fact, be no such witness, in which case $\kappa(K) =\, \bottom$.
Let $\kappa_{l,p}$ be short for $\kappa(J_l \cap I_{l,p})$.
Notice that by Fredman's trick we can deduce the sorted order on
\[
S_{l,p} = \{A(i,k) + B(k,j) \;|\; k \in J_{l} \cap I_{l,p}\}
\]
without additional comparisons, for any $l$ and $p$.  We can therefore compute
the top-level witnesses
$\{\kappa_{\log\log n-1,p}\}_{p\in [n/(g2^{\log\log n -1})]}$ with 
$O(\frac{n}{g2^{\log\log n}}\cdot \log n) = O(\sqrt{n})$ comparisons via binary search.  Our goal now is
to compute the witnesses at all lower level intervals with $O(\sqrt{n})$ comparisons.
Suppose we have computed the level-$(l+1)$ witness $\kappa_{l+1,p}$
and wish to compute the level-$l$ witnesses of the constituent sequences, namely 
$\kappa_{l,2p}$ and $\kappa_{l,2p+1}$.
Define $\kappa_{l,2p}' = \kappa(J_{l+1} \cap I_{l,2p})$.  
Note that $\kappa_{l,2p}'$ is determined by $\kappa_{l+1,p}$ and the sorted order on $S_{l+1,p}$.
The distance between $\kappa_{l,2p}$ and $\kappa_{l,2p}'$ (according to the sorted order on $S_{l,2p}$) 
is stochastically dominated by a geometric random variable with mean 1.\footnote{With probability $1/2$ 
$\kappa_{l,2p} = \kappa_{l,2p}'$; with probability less than 1/4 $\kappa_{l,2p}$ is one less than $\kappa_{l,2p}'$
according to the sorted order on $S_{l,2p}$, and so on.}
The expected number of comparisons needed to determine $\kappa_{l,p}$ using linear search is therefore $O(1)$.
These geometric random variables are independent due to the independence of the samples, so we can apply standard
Chernoff-type concentration bounds~\cite{DubhashiPanconesi09}. The probability that the sum of these 
independent geometric random variables exceeds twice its expectation $\mu$
is $\exp(-\Omega(\mu)) = \exp(-\Omega(\sqrt{n}))$.

Once we have computed all the witnesses for level-0, $\{\kappa_{0,p}\}_{p\in [n/g]}$, we simply have to choose the best among them,
so $C(i,j) = \min \{A(i,\kappa_{0,p}) + B(\kappa_{0,p},j)  \;|\; p\in [n/g]  \mbox{ and }  \kappa_{0,p} \neq \bottom\}$.  
The total number of witnesses computed for fixed $i,j$ is $\sum_{l\ge 0} n/(g2^l) < 2n/g$.
The total number of comparisons is therefore $O(n^2g)$ to sort $D$ and $O(n^3/g)$ 
to compute all the witnesses and $C$, which is $O(n^{5/2})$ when $g=\sqrt{n}$.

To improve the $O(n^3 (\log\log n)^2/\log n)$ algorithm we apply the ideas above with different parameters.
Let $g = \Theta(\log n/\log\log n)$.  We consider the same partitions $\{I_{l,p}\}_{l,p}$ and nested index sets $\{J_l\}_{l}$,
but only use the first $\log\log\log n$ levels, not $\log\log n$ as before.
For each level $l\in [\log\log\log n]$, index $p\in [n/(g2^l)]$, and permutation $\pi \,:\, [g]\rightarrow [g]$, 
we compute those pairs $(i,j)$ for which $\pi$ is the sorting permutation on the elements of 
$J_l \cap I_{l,p}$.  This can be done in time linear in the output size, at most $2n^3/g$, and 
$\sum_{l\in [\log\log\log n]} \frac{n}{g2^l} g!c_\epsilon n^{1+\epsilon}$.
When $\epsilon=1/2$ and $g$ is sufficiently small the time spent computing dominating pairs is $O(n^3/g)$.
Since $g=\Theta(\log n/\log\log n)$ we can encode the sorting permutation of each $J_l\cap I_{l,p}$ in one word
and can answer a variety of queries about these permutations in
$O(1)$ time using $O(n)$-size precomputed tables.

Fix a pair $(i,j)\in [n]^2$.
When finding the top-level witnesses $\{\kappa_{\log\log\log n -1, p}\}_{p \in [2n/(g\log\log n)]}$
we can implement each step of the binary searches in $O(1)$ time using table lookups, for a total of $O(n/g)$ time.
We can also implement each step of the linear searches 
for witnesses $\kappa_{l,2p}$ and $\kappa_{l,2p+1}$ in $O(1)$ time using table lookups.
(In addition to encoding the sorting permutations on $J_{l+1} \cap I_{l+1,p}$, $J_l\cap I_{l,2p}$, and $J_l\cap I_{l,2p+1}$, 
we also need to encode the positions of $J_{l+1} \cap I_{l+1,p}$ within $J_{l} \cap I_{l+1,p}$ as a length-$2g$ bit vector.
This is needed in order to find $\kappa_{l,2p}'$ and $\kappa_{l,2p+1}'$ in $O(1)$ time, given $\kappa_{l+1,p}$ and the sorted order on $S_{l+1,p}$.)
Over all $(i,j) \in [n]^2$ the total number of comparisons is $O(n^3/g) = O(n^3 \log\log n/\log n)$ with high probability.
\end{proof}

The trivial time to solve \ZeroTriangle{} on sparse $m$-edge graphs is $O(m^{3/2})$.
Such graphs contain at most $O(m^{3/2})$ triangles, which can be enumerated in $O(m^{3/2})$ time.
We now restate and prove~\ref{thm:sparse-ZeroTriangle}.\\

\noindent{\bf Theorem~\ref{thm:sparse-ZeroTriangle}.}
{\em \StatementthmsparseZeroTriangle}
\ \\
\begin{proof}
We begin by greedily finding an acyclic orientation of the graph $G=(V,E,w)$.  Iteratively choose the vertex $v$ with the fewest number
of still unoriented edges and direct them all away from $v$.
Since every $m$-edge graph contains a vertex of degree less than $\Delta = \sqrt{2m}$, the maximum outdegree
in this orientation is less than $\Delta$.  We now use $\vec{E}$ instead of $E$ to emphasize that the set is oriented.

Select a random mapping $\colorfunc : V\rightarrow [K]$, where $K$ will be fixed soon.
The expected number of pairs of oriented edges $\{(u,x),(u,x')\}\subset \vec{E}$ having $\colorfunc(x)=\colorfunc(x')$ is less than $m\Delta/K$.
Any coloring that does not exceed this expected value suffices; we do not need to choose $\colorfunc$ at random.
We now sort the set $D$ with $O(m\log m + |D|) = O(m\log m+m\Delta/K)$ comparisons~\cite{Fredman76}, where
\[
D = \{w(u,x) - w(u,x') \;\,|\,\; u\in V \,\mbox{ and }\, (u,x),(u,x')\in\vec{E} \,\mbox{ and }\, \colorfunc(x) = \colorfunc(x')\}.
\]
Call a triangle on $\{u,v,x\}$ {\em type-$((u,v),\kappa)$} if the orientation of the edges is $(u,v),(u,x),(v,x)$ and $\colorfunc(x)=\kappa$.  
Clearly every triangle is of one type and there are $mK$ types.  A type-$((u,v),\kappa)$ zero-weight triangle exists iff $-w(u,v)$ appears in 
the set $\{w(u,x) + w(v,x) \;|\; \mbox{$x\in V$ and $\colorfunc(x) = \kappa$}\}$.
By Fredman's trick the sorted order of this set is determined by the sorted order of $D$, 
since $w(u,x) + w(v,x) < w(u,x') + w(v,x')$ iff $w(u,x) - w(u,x') < w(v,x') - w(v,x)$.  See Figure~\ref{fig:triangle}.
We can therefore determine if there exists a zero-weight triangle of a particular type with $O(\log\Delta)$ comparisons via binary search.
The total number of comparisons is $O(m\log m + m\Delta/K + mK\log\Delta)$, which is $O(m^{5/4}\sqrt{\log m})$ when 
$K=\sqrt{\Delta/\log\Delta}=O(m^{1/4}/\sqrt{\log m})$.  The $\sqrt{\log m}$ factor can be shaved off using randomization, exactly as
in Theorem~\ref{thm:target-min-plus-rand}.  We form $\log\log n$ levels of colorings, where color class $p$ at the $(l+1)$th level is the
union of classes $2p$ and $2p+1$ at the $l$th level.
After the searches are conducted at level $l+1$, the expected cost per search at level $l$ is $O(1)$. 

\begin{figure}
\centering
\scalebox{.25}{\includegraphics{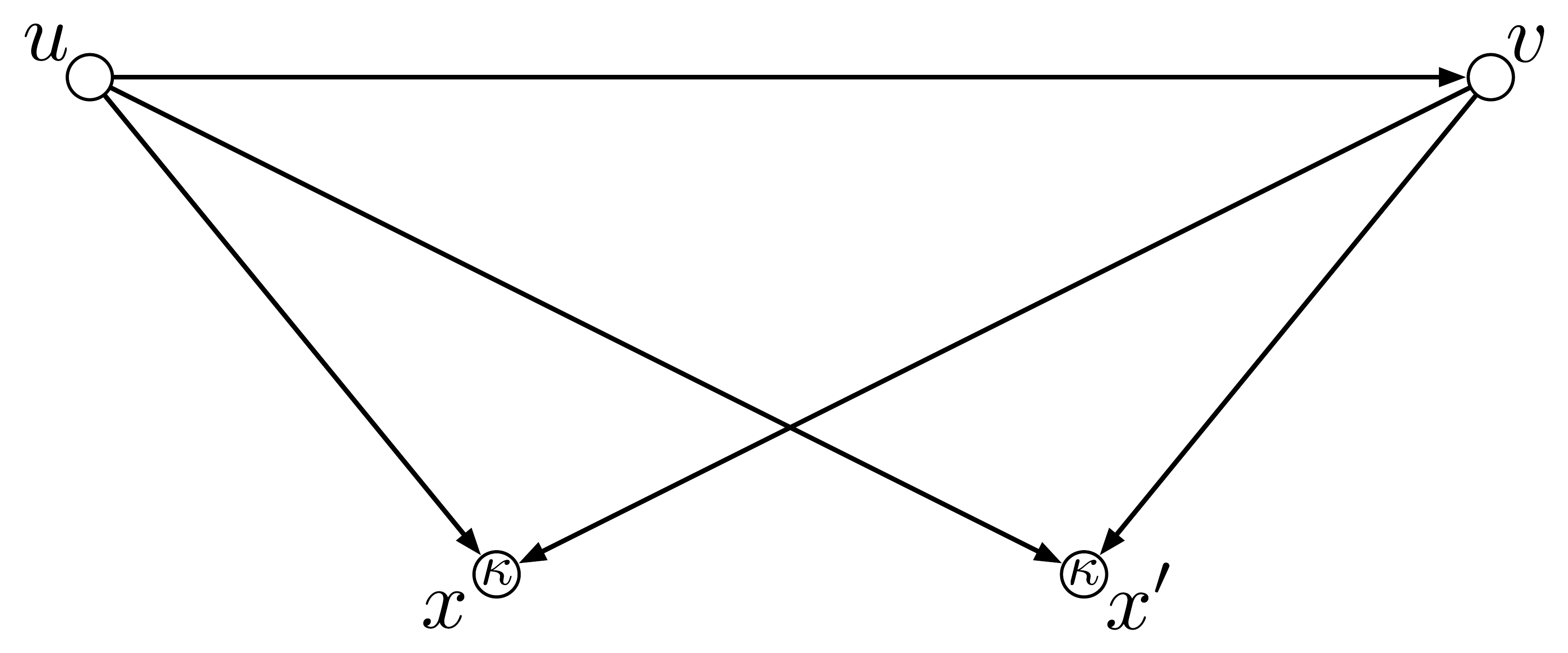}}
\caption{\label{fig:triangle}In this example $\colorfunc(x)=\colorfunc(x')=\kappa$ and both triangles on $\{u,v,x\}$ and $\{u,v,x'\}$
are of the same type.}
\end{figure}

To solve \ZeroTriangle{} efficiently we greedily orient the graph as before, stopping when all remaining vertices have degree at least $\Delta$,
where $\Delta$ is a parameter to be fixed shortly.  (The unoriented subgraph remaining is called the {\em $\Delta$-core}.)
For each vertex $u$ and each pair of outgoing edges $(u,v),(u,x)\in\vec{E}$, we check whether $(u,v,x)$ is a triangle and, if so, whether it has zero weight.  (Note that the edge $(v,x)$, if it exists, may be in the $\Delta$-core and therefore not have an orientation.)  This takes $O(m\Delta)$ time.
It remains to check triangles contained entirely in the $\Delta$-core.  Since the $\Delta$-core has 
at most $2m/\Delta$ vertices we can solve \ZeroTriangle{} on it in $O((m/\Delta)^3 (\log\log m)^2/\log m)$ time
or $O((m/\Delta)^3 \log\log m/\log m)$ time with high probability.
The total cost is balanced when $\Delta = \sqrt{m} \left((\log\log m)^2 / \log m)\right)^{1/4}$
or $\Delta =  \sqrt{m} \left(\log\log m / \log m\right)^{1/4}$ depending on whether uses the randomized 
or deterministic \ZeroTriangle{} algorithm.
\end{proof}

The \ConvolutionThreeSUM{} problem is easily reducible to \ThreeSUM, so our $O(n^{3/2}\sqrt{\log n})$
and $O(n^2 / \polylog(n))$ bounds for \ThreeSUM{}  extend directly to \ConvolutionThreeSUM.
However, \ConvolutionThreeSUM{} has additional structure, which makes it amenable to the same
random sampling techniques used in Theorem~\ref{thm:target-min-plus-rand}.
We give only a sketch of the proof of Theorem~\ref{thm:ConvolutionThreeSUM}
as the analysis is essentially the same as that found in Theorem~\ref{thm:target-min-plus-rand}.\\

\noindent{\bf Theorem~\ref{thm:ConvolutionThreeSUM}.}
{\em \StatementthmConvolutionThreeSUM}
\ \\

\begin{proof} (sketch)
In the \ConvolutionThreeSUM{} problem we must determine if there is a $k\in[n]$ such that $A(k)$ occurs
on the $k$th antidiagonal of the matrix $A+A$.  
In contrast to \ThreeSUM, the rows and columns of $A+A$ are not sorted.  On the other hand, we do not
need to look for $A(k)$ in the whole matrix, just those locations along an antidiagonal.

The $O(n^{3/2}\sqrt{\log n})$ decision tree bound is proved as in Section~\ref{sect:nonuniform}, by partitioning
the matrix into $g\times g$ blocks and for each $k$, conducting binary searches for $A(k)$ in the appropriate
antidiagonals of at most $2n/g$ boxes.  In order to shave off the $\sqrt{\log n}$ factor we use the same random sampling
approach of Theorem~\ref{thm:target-min-plus-rand}.  We partition $A+A$ at $\log\log n$ levels, where level-$l$
boxes have size $g2^l \times g2^{l}$ and are the union of four level-$(l-1)$ boxes.  The rows and columns are sampled
at $\log\log n$ levels, where a row or column at level $l-1$ is promoted to level $l$ with probability 1/2.  Note that an element of $A+A$
appears at level-$l$ if and only if both its row and column are in the level-$l$ sample.  Since elements along any antidiagonal
share no rows or columns, the events that they appear at level-$l$ are entirely independent.
This independence property allows us to search for $A(k)$ in level-$l$ sampled boxes in $O(1)$ expected time,
given the predecessors of $A(k)$ in the level-$(l+1)$ sampled boxes.  Algorithms running in $O(n^2 (\log\log n)^2/\log n)$ (deterministically)
or $O(n^2 \log\log n/\log n)$ (with high probability) are obtained using the methods applied in Theorems~\ref{thm:target-min-plus} and \ref{thm:target-min-plus-rand}.
Alternatively, we could apply the Williams-Williams reduction~\cite{WilliamsW13} from \ConvolutionThreeSUM{} to \ZeroTriangle{}
and then invoke the algorithms of Theorems~\ref{thm:target-min-plus} and \ref{thm:target-min-plus-rand} as black boxes.
\end{proof}

\section{Conclusion}\label{sect:conclusion}

Since the introduction of Fredman's~\cite{Fredman76} $(\min,+)$-product algorithm in 1976,
many have become comfortable with the idea that some numerical problems {\em naturally} have a large gap ($\tilde{\Omega}(\sqrt{n})$)
between their (nonuniform) decision-tree complexity and (uniform) algorithmic complexity.\footnote{Other examples include $(\min,+)$-convolution, $(\operatorname{median},+)$-convolution, polyhedral 3SUM (see Bremner et al.~\cite{BremnerCDEHILPT14}), and 
\Erdos-Szekeres partitioning, that is, decomposing a sequence into $O(\sqrt{n})$ monotonic subsequences.  See Bar-Yehuda and Fogel~\cite{BF98},
Dijkstra~\cite{Dij80}, and Fredman~\cite{F76}.}
From this perspective, our decision trees for \ThreeSUM{} and \ZeroTriangle{} (with depth $\tilde{O}(n^{3/2}$  and $\tilde{O}(n^{5/2})$)
do not constitute convincing evidence that \ThreeSUM{} and \ZeroTriangle{} have truly subquadratic and subcubic algorithms.
However, Williams's~\cite{Williams14} recent breakthrough on the algorithmic complexity of $(\min,+)$-product should shake one's confidence
that these $\sqrt{n}$ gaps are natural.  To close them one may simply need to develop more sophisticated algorithmic machinery.

The exponent $3/2$ has a special significance in \Patrascu's program~\cite{Patrascu10} of 
conditional lower bounds based on hardness of \ThreeSUM.
His superlinear lower bounds on triangle enumeration and polynomial lower bounds on dynamic data structures
depend on the complexity of \ThreeSUM{} being $\Omega(n^{3/2 + \epsilon})$, for some $\epsilon>0$.
In most other \ThreeSUM-hardness proofs there is nothing sacred about the 3/2 threshold (or any other exponent).
For example, if \ThreeSUM{} requires $\Omega(n^{1.05})$ time
then finding three collinear points in a set $P\subset \mathbb{R}^2$ also requires $\Omega(|P|^{1.05})$ time~\cite{GajentaanO95}.

\ignore{
\bibliographystyle{plain}
\bibliography{../../references}
}

\appendix

\section{Bichromatic Dominating Pairs}\label{sect:redblue}

For the sake of completeness we shall review a standard divide and conquer dominating pairs
algorithm of Preparata and Shamos~\cite[p.~366]{PreparataShamos85} and
give a short proof of Lemma~\ref{lem:redblue} due to Chan~\cite{Chan08}.

\subsection{The Divide and Conquer Algorithm}

We are given $n$ red and blue points in $P\subset \mathbb{R}^d$, at least one of each color, 
and wish to report all pairs $(p,q)$ where 
$p=(p_i)_{i\in [d]}$ is red, $q=(q_i)_{i\in [d]}$ is blue and 
$p_i \ge q_i$ for each $i\in [d]$.
When $d=0$ the algorithm simply reports every pair of points, so assume $d\ge 1$.
Find the median $h$ on the last coordinate in $O(n)$ time~\cite{BFPRT73} 
and partition $P$ into disjoint sets $P_L$, $P_R$ 
of size at most $\ceil{n/2}$, where
\begin{align*}
P_L &\subset \{p\in P \:|\: p_{d-1} \le h\}\\
P_R &\subset \{p\in P \:|\: p_{d-1} \ge h\}.
\end{align*}
Furthermore, there cannot be a red $p\in P_L $  and blue $q\in P_R$ such that $p_{d-1} = q_{d-1} = h$.\footnote{In other words, 
among points with the same last coordinate, blue points precede red points.
If the domination criterion were {\em strict}, that is, if $(p,q)$ were a dominating pair only if $p_i > q_i$ for all $i\in [d]$,
then we would break ties the other way, letting red points precede blue points.}
At this point all dominating pairs are in $P_L,$ or $P_R$, or have one point in each, in which case the
blue point is necessarily in $P_L$ and the red in $P_R$.  We make three recursive calls to find dominating
pairs of each variety.  The first two calls are on $\ceil{n/2}$ points in $\mathbb{R}^d$.  The third recursive call
is on all blue points in $P_L$ and all red points in $P_R$; after stripping their last coordinate they lie in $\mathbb{R}^{d-1}$.

Excluding the cost of reporting the output, the running time of this algorithm is bounded by $T_d(n)$, defined inductively as
\begin{align*}
T_0(n) &= T_d(1) \,=\, 0\\
T_d(n) &= 2T_d(n/2) + T_{d-1}(n) + n
\end{align*}

We prove by induction that $T_d(n) \le c_\epsilon n^{1+\epsilon} - n$, a bound which holds in all base cases.  
Assuming the claim holds for all smaller values of $d$ and $n$,
\begin{align*}
T_d(n) &\le 2\left(c_\epsilon^d (n/2)^{1+\epsilon} - n/2\right) + \left(c_\epsilon^{d-1}n^{1+\epsilon} - n\right) + n\\
	&= \left(c_\epsilon^d/2^{\epsilon} + c_{\epsilon}^{d-1}\right) n^{1+\epsilon} - n\\
	&= \left(1/2^{\epsilon} + 1/c_{\epsilon}\right) \cdot c_\epsilon^d n^{1+\epsilon} - n\\
	&= c_\epsilon^d n^{1+\epsilon} - n	& \mbox{By defn. of $c_\epsilon = 2^{\epsilon}/(2^{\epsilon}-1)$.}
\end{align*}

\end{document}